\documentclass[11pt]{article}
\title{\textbf{\Large Limit order books, diffusion approximations and reflected SPDEs: from microscopic to macroscopic models}}
\author{Ben Hambly\footnote{ \href{mailto:ben.hambly@maths.ox.ac.uk}{ben.hambly@maths.ox.ac.uk}}\hspace{2mm}, \hspace{1mm} Jasdeep Kalsi\footnote{\href{mailto:jamespnewbury@gmail.com}{jasdeep.kalsi@maths.ox.ac.uk}} \hspace{2mm} and James Newbury \footnote{\href{mailto:}{james.newbury@merton.oxon.org}}\\\\
Mathematical Institute, University of Oxford\\\\}
\date{\today}

\usepackage{amsmath, amssymb, amsthm,bbm}
\usepackage{mathrsfs}
\usepackage{enumerate}
\newtheorem{thm}{Theorem}[section]

\newtheorem{cor}[thm]{Corollary}
\newtheorem{defn}[thm]{Definition}
\newtheorem{eg}[thm]{Example}
\newtheorem{prop}[thm]{Proposition}
\newtheorem{rem}[thm]{Remark}

\usepackage[margin=1in]{geometry}
\usepackage{graphicx}
\usepackage{epstopdf}
\usepackage{hyperref}
\usepackage{amsmath}
\usepackage{amsfonts}
\usepackage{setspace}
\usepackage{subcaption}
\newlength{\bibitemsep}
\newlength{\bibparskip}
\let\oldthebibliography\thebibliography
\renewcommand\thebibliography[1]{%
  \oldthebibliography{#1}%
  \setlength{\parskip}{\bibitemsep}%
  \setlength{\itemsep}{\bibparskip}%
}
\usepackage{indentfirst}
\usepackage{eqnarray}
\usepackage{enumerate}
\usepackage{commath}
\usepackage{amssymb}
\usepackage{afterpage}
\usepackage{bbm}
\usepackage{algorithm}
\usepackage{algpseudocode}
\usepackage{float}

\numberwithin{equation}{section}
\makeatletter

\newcommand{\xRightarrow}[2][]{\ext@arrow 0359\Rightarrowfill@{#1}{#2}}
\makeatother
\allowdisplaybreaks[1]

\begin{document}

\maketitle
\begin{abstract}
Motivated by a zero-intelligence approach, the aim of this paper is to connect the microscopic (discrete price and volume), mesoscopic (discrete price 
and continuous volume) and macroscopic (continuous price and volume) frameworks for the modelling of limit order books, with a view to providing a 
natural probabilistic description of their behaviour in a high to ultra high-frequency setting. Starting with a microscopic framework, we first examine the limiting 
behaviour of the order book process when order arrival and cancellation rates are sent to infinity and when volumes are considered to be of infinitesimal 
size. We then consider the transition between this mesoscopic model and a macroscopic model for the limit order book, obtained by letting the tick size tend to zero. 
The macroscopic limit can then be described using reflected SPDEs which typically arise in stochastic interface models. We then use financial data to 
discuss a possible calibration procedure for the model and illustrate numerically how it can reproduce observed behaviour of prices. This could then 
be used as a market simulator for short-term price prediction or for testing optimal execution strategies.
\end{abstract}


\section{Introduction}\label{Introduction}

The rising prevalence of order-driven markets in recent years has generated a significant interest in 
the modelling of limit order books, for an overview see the survey paper \cite{GP}. 
In such markets, three specific types of orders can be submitted. 
Firstly, limit orders are orders to buy or sell a designated number of shares at a specified 
price or better. Secondly, market orders are orders to immediately buy or sell a certain number of 
shares at the best available price. Finally, cancellation orders enable a market participant to 
cancel an existing limit order. Whilst market orders are instantly matched against the best available 
limit orders of the opposite quote, the collection of unexecuted and uncancelled limit orders is 
recorded in the limit order book (LOB), according to price and time priority. In limit order book 
terminology, the bid refers to the price of the best limit buy order, whereas the ask designates 
the price of the best sell order. The average of the bid and the ask is referred to as the mid. 
Two other quantities of interest are the spread, which corresponds to the difference between the ask 
and the bid, and the tick size, which is the smallest price increment in the market. 

There is now a large amount of high-quality financial time series of orders enabling one to conduct 
statistical analyses of various limit order book features and to guide the development of models. 
A good model for the order book should describe the evolution of prices and capture some of 
the stylised facts that are observed in financial time series. It can then be used for examining strategies 
for order placement and for the optimal execution of large orders. 

Limit order book models have been developed by two independent schools of 
thought. The first one, initiated by economists, has been based on a `perfect-rationality' approach, where 
market participants all employ optimal strategies to place limit orders.
The second one, led by econophysicists and mathematicians, has been associated with a `zero-intelligence' 
framework, that is limit order arrivals can be viewed as a purely random process and no strategic
order placement is taken into account.

Within the realm of perfect-rationality, where order flow is considered as static, the central 
issue concerns the strategic trading decisions of agents in which they maximise 
their individual utility. Notable models in the perfect-rationality literature include those due to 
Mendelson~\cite{M}, who analysed the statistical behaviour of the market from a clearing house
perspective, Kyle~\cite{AK}, where the question of insider trading with sequential auctions is 
addressed, Ro\c{s}u~\cite{R}, who introduced the notion of optimal choice between market orders 
and limit orders, and Almgren and Chriss~\cite{AC}, where a model for the optimal execution of large 
orders is developed.

In the zero-intelligence approach the order flow is treated as dynamic and 
the focus is shifted to the random nature of order arrivals. One of the first models dealing with 
this was developed by Kruk~\cite{K}, where he established a functional limit theorem for the order 
flow in a continuous double-auction setting. More recently, there has been a significant interest 
in modelling the book as a multiclass queueing system, \cite{AJ},~\cite{BC},~\cite{CS},~\cite{MT},~\cite{HL}. 
In order to deal with a market where orders are submitted at 
high frequency, Cont and Larrard \cite{CL2} considered a heavy traffic approximation of the 
order book process from a queueing theory perspective. The order book is reduced to the 
best bid and ask queues: once the bid (or ask) 
queue has been depleted, it takes a new value drawn from a stationary distribution representing 
the depth of the order book after a price change. 

Over the years, there have also been a number of attempts to 
establish scaling limits for the full order book.  One of the first papers to explore this direction is \cite{BoCe}, where
the order book is modelled as a two-species interacting 
particle system and a hydrodynamic limit is obtained for the associated empirical process. This 
particle system approach is also used by \cite{DG}, where their limit is actually 
an ODE with a constant price.
The work of \cite{LRS} obtains a limit for one side of the order book as a measure-valued process. In~\cite{BHQ} 
and~\cite{HP} Horst et al. establish functional limit theorems for two-sided order books and obtain 
PDE or SPDE limits depending on the initial scaling procedure. Finally, Zheng~\cite{Zheng} 
(using the results of Kim et al.~\cite{KSZ} on stochastic Stefan problems) and M\"{u}ller~\cite{Mu} 
and \cite{KRM} developed models for the order book as a stochastic free boundary problem.

\subsection{From micro to macro models}

Motivated by a zero-intelligence approach, the aim of this paper is to bridge the gap between \textbf{microscopic} 
(discrete volume and price), \textbf{mesoscopic} (continuous volume and discrete price) and \textbf{macroscopic} 
(continuous volume and price) models of limit order books. The financial context of our study is the following: we 
consider an order-driven market where orders and cancellations are submitted at very high frequency. Starting with 
a discrete-space model describing the microscopic evolution of the order book, we prove that by sending order arrival 
and cancellation rates to infinity and by rescaling order volumes, the behaviour of the book can be described in terms 
of a system of coupled stochastic differential equations. This is what we call the mesoscopic limiting process. Next, by sending 
the tick size to zero, we derive a macroscopic SPDE limit from the mesoscopic process.

Even though we send the tick size to zero we wish to capture the fact that in a high frequency trading environment 
the price changes are comparatively rare in the evolution of the order book. Thus we will consider our model for the book 
as generating price changes at a much lower frequency, so there is a natural separation of time scales. Our price changes 
will be a macroscopic tick movement occurring as a result of imbalances created in the book by the order flow. 
Our main mathematical result is Theorem~\ref{mesodynamic} showing that the queueing system converges weakly to 
a reflected SPDE for the dynamics of the order book along with a discrete price process evolving in a realistic way.

An outline of the paper is as follows. In the next section we will develop our microscopic model. We allow quite a degree of 
flexibility in the arrival rates and cancellation rates of orders and show in Theorem~\ref{microstatic}, that by letting the 
volume size of orders go to 0 as their rate of arrival goes to infinity that the microscopic system has a scaling limit which 
is a coupled system of stochastic differential equations. The initial result is for the static order book and we then show how 
to incorporate price changes which are functions of the two sides of the order book. In Section 3 we let the tick size go to 
zero and show that this system converges weakly to a reflected SPDE. In our main result (Theorem~\ref{mesodynamic}) 
we also incorporate price changes to get a full model for the order book and the discrete price dynamics it generates.
In Section~4 we illustrate with some examples how our general framework can incorporate some natural models. 
In Section~5 we develop the numerical application for our framework. By considering data from LOBSTER we show how 
to determine some of the parameters in a simple version of the model and how it will produce realistic order book profiles 
and prices series. The proofs of the Theorems are then given in the Appendix.

\section{Diffusion Approximations: From Microscopic to Mesoscopic Models}

We begin this section by introducing a simple microscopic order book model, where order arrivals and cancellations are driven by Poisson processes. 
One could also use more complicated point processes as the basis for the model, for instance Hawkes processes. We choose Poisson processes 
both to simplify the analysis, and since such dynamics can be considered an approximation to a large class of point processes which are natural in this context. Other papers which assume Poisson driven order book dynamics include \cite{AJ}, \cite{BHQ}, \cite{CS}, \cite{DG}, \cite{LRS} and \cite{MT}. 
The rates for order arrivals, market orders and cancellations will be allowed to depend on the current mid price, the price relative to the mid and the 
number of orders currently on the book at the point in question. 

We will first consider a model for the evolution of the book in between price changes, 
and will then add price changes by introducing stopping times which depend on how the book has evolved. This will allow us to easily maintain a 
separation of time-scales between the evolution of the order book profiles and the corresponding price process when passing to mesoscopic and 
macroscopic limits, where we will scale time and space for our models. Maintaining this separation is sensible since, typically, price movements 
occur on a significantly slower time-scale to order book events. For example, examining order book data for the SPDR Trust Series I from June 21 2012 
between 11:00am and 12:00pm, we see that there were 840549 order book events and approximately 2453 price changes.

Once our microscopic model has been described, our goal is then to establish a diffusion approximation for the model. This is done initially for the static model, where price changes are not included, and then for the dynamic case. 

%
%

\subsection{The Discrete Order Book Process in a Static Setting}

In this section we describe the dynamics of our microscopic order book model when it is static i.e. we describe its behaviour in between price changes. The index $n$ here will be used later in order to take a diffusive scaling of the model. 

We will work on a relative price grid in that, at any given time, the $i^{\textrm{th}}$ price point of the bid/ask side of the book refers to the price which is $i$ ticks away from the best bid/ask respectively. The grid is given by $\left\{0,1,2,...,N \right\}$ for some $N \in \mathbb{N}$ . For every $n \in \mathbb{N}$, we consider two $N-1$- dimensional processes, $Z_n^b=(Z_n^{b,1},...Z_n^{b,N-1})$ and $Z_n^a=(Z_n^{a,1},...Z_n^{a,N-1})$, each taking values in $\mathbb{Z}^{N-1}$ and representing the limit order volumes currently on the bid and ask sides respectively of the static discrete order book process. The mid price $m \in \mathbb{R}$ is taken to be fixed here, and by convention we think of the spread as being constantly equal to two ticks. For each $i \in \left\{1,2,...,N-1 \right\}$, $Z_n^{b,i}$ then represents the number of outstanding orders to buy at price $m-i$, whilst $Z_n^{a,i}$ represents the number of outstanding limit orders to sell at price $m+i$. Order and cancellation sizes are assumed to be 1 in the static setting, although the results here can easily be adapted to the case where order sizes are assumed only to be bounded. We choose the rates at which different orders arrive in our model such that they possess the following three features.

\begin{enumerate}
\item At each price level, there is a common high frequency rate for limit orders and cancellation/market orders. We allow for these rates to be dependent on the relative price (with respect to the mid), the current position of the mid and the number of offers currently at that price. These terms are intended to capture the effects of high frequency trading.
\item Residual imbalance between limit orders and cancellations/market orders at different price levels gives lower frequency terms. These are once again allowed to be dependent on price, the current midprice and the number of offers at that price. 
\item Orders undergo a lower frequency random walk. This is intended to capture the effects of traders repositioning their offers in the book, as many of the cancellations that occur will be quickly followed by a limit order at an adjacent queue. We assume that each order moves to a neighbouring queue at a certain rate. This will have a smoothing effect on the profile of the order book.

\end{enumerate}
Altogether, this motivates the following description for the dynamics of the bid side of the order book in our model.  To simplify the notation here, we define  $e_i$ to be the usual basis functions for $\mathbb{R}^{N-1}$ for $i=1,2,...,N-1$, and we use the convention $e_0=e_N=0$.

\begin{enumerate}[(i)]
\item For $i \in \left\{1,...,N-1 \right\}$, $Z_n^b \rightarrow Z_n^b + e_i$ at exponential rate 
\begin{equation*}\begin{split}
\frac{1}{2} & \sigma_{b,m,n}^2 (i, Z_n^{b,i})\left( 1+ \mathbbm{1}_{\left\{Z_n^{b,i}=0\right\}} \right) + f_{b,m,n}\left(i, Z_n^{b,i} \right).
\end{split}
\end{equation*}
\item For $i \in \left\{1,...,N-1 \right\}$, $Z_n^b \rightarrow Z_n^b - e_i$ at exponential rate 
\begin{equation*}
\begin{split}
\frac{1}{2} &  \sigma_{b,m,n}^2 (i, Z_n^{b,i}) \mathbbm{1}_{\left\{Z_n^{b,i} \geq 1 \right\}} + g_{b,m,n}\left(i, Z_n^{b,i} \right)\mathbbm{1}_{\left\{Z_n^{b,i} \geq 1 \right\}}.
\end{split}
\end{equation*}
\item  For $i \in \left\{1,...,N-1 \right\}$, $Z_n^b \rightarrow Z_n^b + e_{i-1}-  e_i$ at exponential rate 
\begin{equation*}
\alpha_{b,n}Z_n^{b,i}
\end{equation*}
\item  For $i \in \left\{1,...,N-1 \right\}$, $Z_n^b \rightarrow Z_n^b + e_{i+1}-  e_i$ at exponential rate 
\begin{equation*}
\alpha_{b,n}Z_n^{b,i}
\end{equation*}
\end{enumerate}
Similarly, the dynamics of the ask side of the book are given by:
\begin{enumerate}[(i)]
\item For $i \in \left\{1,...,N-1 \right\}$, $Z_n^a \rightarrow Z_n^a + e_i$ at exponential rate 
\begin{equation*}\begin{split}
\frac{1}{2} & \sigma_{a,m,n}^2 (i, Z_n^{a,i})\left( 1+ \mathbbm{1}_{\left\{Z_n^{a,i}=0\right\}} \right) + f_{a,m,n}\left(i, Z_n^{a,i} \right).
\end{split}
\end{equation*}
\item For $i \in \left\{1,...,N-1 \right\}$, $Z_n^a \rightarrow Z_n^a - e_i$ at exponential rate 
\begin{equation*}
\begin{split}
\frac{1}{2} &  \sigma_{a,m,n}^2 (i, Z_n^{a,i}) \mathbbm{1}_{\left\{Z_n^{a,i} \geq 1 \right\}} + g_{a,m,n}\left(i, Z_n^{a,i} \right)\mathbbm{1}_{\left\{Z_n^{a,i} \geq 1 \right\}}.
\end{split}
\end{equation*}
\item  For $i \in \left\{1,...,N-1 \right\}$, $Z_n^a \rightarrow Z_n^a+e_{i-1}-  e_i$ at exponential rate 
\begin{equation*}
\alpha_{a,n}Z_n^{a,i}
\end{equation*}
\item  For $i \in \left\{1,...,N-1 \right\}$, $Z_n^a \rightarrow Z_n^a +e_{i+1}-  e_i$ at exponential rate 
\begin{equation*}
\alpha_{a,n}Z_n^{a,i}
\end{equation*}
\end{enumerate}

In the above we have that: 
\begin{enumerate}[(a)]
\item For $k \in \{b,a\}$, every $n$ and every $m \in \mathbb{R}$, $\sigma_{k,m,n}$ is a map from $\left\{1,2,...,N-1\right\} \times \mathbb{N} \rightarrow \mathbb{R}^+$.
\item For $k \in \{b,a\}$, every $n$ and every $m \in \mathbb{R}$, $f_{k,m,n}$ and $g_{k,m,n}$ are maps from $\left\{1,2,...,N-1\right\} \times \mathbb{N} \rightarrow \mathbb{R}^+.$ 
\end{enumerate}

\begin{rem} {\rm We remark here that market orders have the same impact on the profile of the book as cancellations at the best price levels. Market orders are therefore accounted for in these static dynamics.} 
\end{rem}

\begin{rem}{\rm 
Our model here only accounts for placement of small orders, which we have taken without loss of generality to be of size one. In Section~4 we give an example of how to include larger orders on a longer timescale in our dynamic model.}
\end{rem}

\subsection{Heavy Traffic Diffusion Approximation in a Static Setting}

We now switch our attention to the heavy traffic approximation of the  suitably rescaled static microscopic order book process. Time is accelerated by a factor of $n$ and volumes are divided by $\sqrt{n}$. Therefore, we are considering the limits of the processes $$\tilde{Z}_n^b(t):=\frac{Z_n^b(nt)}{\sqrt{n}} \; \; \; \; \; \; \textrm{and} \; \; \; \; \; \; \tilde{Z}_n^a(t):=\frac{Z_n^a(nt)}{\sqrt{n}}.$$
These processes therefore take values in $\frac{1}{\sqrt{n}} \mathbb{N}^{N-1}$ and the limiting process will take values in $[0, \infty)^{N-1}$. In order to obtain convergence, we need that various quantities in our microscopic model converge suitably. We will assume that
\begin{enumerate}[(i)]
\item For $k \in \{b,a\}$, $m \in \mathbb{R}$, $i \in \left\{1,2,...,N-1 \right\}$, $u \in \mathbb{N}$ and $n \geq 1$ $$\sigma_{k,m,n}(i,u)= \sigma_{k,m}\left(i,\frac{u}{\sqrt{n}} \right).$$
\item For $k \in \{b,a\}$, $m \in \mathbb{R}$, $i \in \left\{1,2,...,N-1 \right\}$, $u \in \mathbb{N}$ and $n \geq 1$ $$f_{k,m,n}(i,u)= \frac{1}{\sqrt{n}}f_{k,m}\left(i,\frac{u}{\sqrt{n}} \right),$$ $$g_{k,m,n}(i,u)= \frac{1}{\sqrt{n}}g_{k,m}\left(i,\frac{u}{\sqrt{n}} \right).$$ 
\item For $k \in \{b,a\}$ and $n \geq 1$, $\alpha_{k,n}= \frac{1}{n} \alpha_k > 0$.
\item For $k \in \{b,a\}$, $\frac{Z_n^k(0)}{\sqrt{n}} \implies X^k(0)$ in law in $[0, \infty)^{N-1}$ as $n \rightarrow \infty$.
\end{enumerate} 
Here, the functions $\sigma_{k,m}$, $f_{k,m}$ and  $g_{k,m}$ are all measurable from $\left\{1,2,...,N-1 \right\} \times \mathbb{R}^+ \rightarrow \mathbb{R}^+$. For technical reasons, we further assume that these functions are Lipschitz continuous in the second argument. Note that this implies boundedness on compact sets.
\begin{thm}
\label{microstatic}
The $\frac{1}{\sqrt{n}}\mathbb{N} \times \frac{1}{\sqrt{n}}\mathbb{N}$- valued process $(\tilde{Z}_n^b, \tilde{Z}_n^a)$ converges weakly in \\$\mathscr{M}\left(\mathbb{D}([0, \infty) ; \mathbb{R}^{N-1}) \times \mathbb{D}([0, \infty) ; \mathbb{R}^{N-1})\right)$ as $n \rightarrow \infty$ to the unique $[0, \infty)^{N-1} \times [0, \infty)^{N-1}$-valued strong Markov diffusion process $(X^b,X^a)$ which satisfies the following system of reflected SDEs:
$$\textrm{d}X_t^{b,i}= \alpha_b(X_t^{b,i+1}+X_t^{b,i-1}-2X_t^{b,i}) \textrm{dt}+h_{b,m}(i,X_t^{b,i}) \textrm{dt} + \sigma_{b,m}(i,X_t^{b,i}) \textrm{d}W_t^{b,i}+ \textrm{d}\eta^{b,i}_t,$$
$$\textrm{d}X_t^{a,i}= \alpha_a(X_t^{a,i+1}+X_t^{a,i-1}-2X_t^{a,i}) \textrm{dt}+h_{a,m}(i,X_t^{a,i}) \textrm{dt} + \sigma_{a,m}(i,X_t^{a,i}) \textrm{d}W_t^{a,i}+ \textrm{d}\eta^{a,i}_t,$$
for $i=1,...,N-1$ with the pinning conditions that $X^{k,0}=X^{k,N}=0$, where $W^{k,i}$ are independent Brownian motions. The $\eta^{k,i}$ are reflection measures which maintain positivity of the $X^{k,i}$.
\end{thm}

\subsection{Dynamic Discrete Order Book Process}

We now describe the mechanism for price movements in the model to give our microscopic dynamic model. Price changes in both directions will be assumed to occur at positive rates which depend on the state of the book at any given time (including the current position of the mid). Our motivating example is the case where these rates are dependent on the imbalance of the number of bid limit orders compared to ask limit orders currently on the book near the mid. Here, relatively more offers to buy near the mid make a price increase more likely and relatively more offers to sell near the mid make a price decrease more likely. 

In order to formalise this, we introduce functions $\theta_{u,m}^n$ and $\theta_{d,m}^n$ for $n \geq 1$ and $m \in \mathbb{R}$, which map $\mathbb{N}^{N-1} \times \mathbb{N}^{N-1} \rightarrow \mathbb{R}_{>0}$. These will determine the rate of upward and downward price movements respectively as a function of the profiles of the bid and ask sides of the book. We also define for $n \geq 1$ functions $R^n$. These map $\mathbb{N}^{N-1} \times \mathbb{N}^{N-1} \times \{u,d\}  \rightarrow \mathscr{M}(\mathbb{N}^{N-1} \times \mathbb{N}^{N-1})$ and will determine the distribution of the new profiles of the bid and ask sides of the book following price changes as a function of the profiles at the time of the price change and the direction of the price change. We fix some $\epsilon >0$, which determines the size of price changes. We additionally introduce i.i.d. rate one exponential random variables $(Y_{n,u}^i)_{i=1}^{\infty}$ and $(Y_{n,d}^i)_{i=1}^{\infty}$ which will be used in the construction. These are independent of each other and of the other driving Poisson processes in the model. With this is place, we can start to construct our dynamic process. Let $Z_{n,1}^b(0), Z_{n,1}^a(0) \in \mathbb{N}^{N-1}$ be the initial profiles for the bid and ask sides of the book respectively, and let $m_n^1$ be the initial mid for the $n^{\textrm{th}}$ microscopic order book. We denote by $Z_{n,1}^b, Z_{n,1}^a$ the processes evolving according to our dynamics for the static microscopic order book with mid $m^1_n$ and initial profiles $Z_{n,1}^b(0), Z_{n,1}^a(0)$. We define the following stopping times.

\begin{equation*}
\tau_{n,u}^1:= \inf \left\{t \geq 0 \; \bigg| \; \int_0^t \theta_{u,m^1_n}^n(Z_{n,1}^b(s),Z_{n,1}^a(s))  \; \textrm{d}s \geq Y_{n,u}^1 \right\},
\end{equation*} 
\begin{equation*}
\tau_{n,d}^1:= \inf \left\{t \geq 0 \; \bigg| \; \int_0^t \theta_{d,m^1_n}^n(Z_{n,1}^b(s),Z_{n,1}^a(s))  \; \textrm{d}s \geq Y_{n,d}^1 \right\}.
\end{equation*} 
Therefore, $\tau_{n,u}^1$ and $\tau_{n,d}^1$ are exponential waiting times, whose arrival rates at time $t$ are given by $\theta_{u,m^1_n}^n(Z_{n,1}^b(s),Z_{n,1}^a(s))$ and $\theta_{d,m^1_n}^n(Z_{n,1}^b(s),Z_{n,1}^a(s))$ respectively. We define $\tau_n^1:= \tau_{n,u}^1 \wedge \tau_{n,d}^1.$ The time $\tau_n^1$ triggers a price change. If $\tau_n^1$= $\tau_{n,u}^1$, we have an upward price change, and set $m^2_n=m^1_n +\epsilon$. Similarly, if $\tau_n^1$= $\tau_{n,d}^1$, we have an downward price change, and set $m^2_n=m^1_n -\epsilon$. We then let $Z_{n,2}^b, Z_{n,2}^a$ be new processes which follow the dynamics of the bid and ask sides of the static microscopic order book with mid $m_n^2$. The initial profile has the law of $R^n(Z_{n,1}^b(\tau^1_n), Z_{n,1}^a(\tau^1_n),u)$ if the price change was upward and $R^n(Z_{n,1}^b(\tau^1_n), Z_{n,1}^a(\tau^1_n),d)$ if the price change was downward. The processes $Z_{n,2}^b$, $Z_{n,2}^a$ are taken to be conditionally independent of the past of the dynamic order book given $(Z_{n,1}^b(\tau_n^1), Z_{n,1}^a(\tau_n^1))$, the direction of the price change, and $m_n^2$. We can iterate this procedure to define further stopping times, price points and processes $Z_{n,i}^b,Z_{n,i}^a$, describing the dynamics after the $i^{\textrm{th}}$ price change. Having defined $(Z_{n,i}^b)_{i=1}^{M}$, $(Z_{n,i}^a)_{i=1}^{M}$, $(\tau^i_{n})_{i=1}^{M-1}$ and $(m^i_n)_{i=1}^{M}$, we set 
\begin{equation*}
\tau_{n,u}^{M}:= \inf \left\{t \geq 0 \; \bigg| \; \int_0^t \theta_{u,m^{M}_n}^n(Z_{n,M}^b(s),Z_{n,M}^a(s))  \; \textrm{d}s \geq Y_{n,u}^{M} \right\},
\end{equation*} 
\begin{equation*}
\tau_{n,d}^{M}:= \inf \left\{t \geq 0 \; \bigg| \; \int_0^t \theta_{d,m^{M}_n}^n(Z_{n,M}^b(s),Z_{n,M}^a(s))  \; \textrm{d}s \geq Y_{n,d}^{M} \right\}.
\end{equation*} 
As before, we define $\tau_n^{M}:= \tau_{n,u}^{M} \wedge \tau_{n,d}^{M}$, with the time $\tau_n^M$ triggering a price change. If $\tau_n^{M}$= $\tau_{n,u}^{M}$, we set $m^{M+1}_n=m^M_n +\epsilon$ and if $\tau_n^{M}$= $\tau_{n,d}^{M}$, we set $m^{M+1}_n=m^M_n -\epsilon$. We then let $Z_{n,M+1}^b, Z_{n,M+1}^a$ be new processes which follows the dynamics of the bid and ask sides of the static microscopic order book with the mid now given by $m^{M+1}_n$. The initial profile has the law of $R^n(Z_{n,M}^b(\tau^M_n), Z_{n,M}^a(\tau^M_n),u)$ if the price change was upward and $R^n(Z_{n,M}^b(\tau^M_n), Z_{n,M}^a(\tau^M_n),d)$ if the price change was downward. Once again, $Z_{n,M+1}^b$, $Z_{n,M+1}^a$ are taken to be conditionally independent of the past of the dynamic order book given $(Z_{n,M}^b(\tau_n^{M}), Z_{n,M}^a(\tau_n^M))$, the direction of the previous price change, and the new mid position. Our dynamic microscopic order book is then described by the processes $(\hat{Z}_n^b(t), \hat{Z}_n^a(t), m_n(t))$ describing the evolution of the two sides of the book and the mid through time, where

\begin{equation*}
\hat{Z}_n^b(t):= \sum\limits_{i=1}^{\infty} Z_{n,i}^b\left(t- \sum\limits_{j=1}^{i-1} \tau_n^j \right) \mathbbm{1}_{\left\{ \sum\limits_{j=1}^{i-1} \tau_n^j \leq t < \sum\limits_{j=1}^{i} \tau_n^j \right\}},
\end{equation*}
\begin{equation*}
\hat{Z}_n^a(t):= \sum\limits_{i=1}^{\infty} Z_{n,i}^a\left(t- \sum\limits_{j=1}^{i-1} \tau_n^j \right) \mathbbm{1}_{\left\{ \sum\limits_{j=1}^{i-1} \tau_n^j \leq t < \sum\limits_{j=1}^{i} \tau_n^j \right\}},
\end{equation*}
\begin{equation*}
m_n(t):=\sum\limits_{i=1}^{\infty} m^i_n \mathbbm{1}_{\left\{ \sum\limits_{j=1}^{i-1} \tau_n^j \leq t < \sum\limits_{j=1}^{i} \tau_n^j \right\}}. 
\end{equation*}
  
\subsection{Heavy Traffic Diffusion Approximation in a Dynamic Setting}

We now present the convergence of our dynamic microscopic model to a dynamic mesoscopic model. We should first, of course, define the dynamic mesoscopic model. This is done in essentially the same way as in the microscopic case. We therefore only give an overview here and refer to the previous section for the precise details. The static mesoscopic model determines the behaviour of the dynamic model in between price changes. The functions determining rates of upward price changes and downward price changes as functions of the book profile are now denoted by $\theta_{u,m}$ and $\theta_{d,m}$ respectively. They are maps from $(\mathbb{R}^+)^{N-1} \times (\mathbb{R}^+)^{N-1} \rightarrow \mathbb{R}_{>0}$, which we assume to be uniformly bounded over $m$ and such that there exists $c>0$ with $\theta_{u,m}, \theta_{d,m} \geq c$. We further assume that these maps are continuous. Our regenerative function which determines the distribution of the new profile following a price change is denoted by $R$ and now maps $(\mathbb{R}^+)^{N-1} \times (\mathbb{R}^+)^{N-1} \times \{u,d\} \rightarrow \mathscr{M}((\mathbb{R}^+)^{N-1} \times (\mathbb{R}^+)^{N-1} )$. This function is also assumed to be continuous, where $\mathscr{M}((\mathbb{R}^+)^{N-1}\times (\mathbb{R}^+)^{N-1} )$ is equipped with the topology of weak convergence. For ease of notation, we introduce the maps $P_n: \mathbb{N}^{N-1} \times \mathbb{N}^{N-1} \rightarrow (\mathbb{R}^+)^{N-1} \times (\mathbb{R}^+)^{N-1}$, which simply divide each coordinate by $\sqrt{n}$. The functions here are then approximated by the corresponding functions in the microscopic model in the following ways.

\begin{enumerate}[(i)]
\item For $k \in \{u,d\}$, $v_1,v_2 \in \mathbb{N}^{N-1}$ and $u_1,u_2 \in (\mathbb{R}^+)^{N-1}$, \begin{equation*}\label{microthetaassumption} \left|n \theta^n_{k,m}(v_1,v_2) - \theta_{k,m}\left(u_1,u_2 \right)\right| \leq r(\| P_n((v_1,v_2)) -(u_1,u_2) \|),\end{equation*} where $\lim\limits_{x \rightarrow 0} r(x) =0$. 
\item For $k \in \{u,d \}$, if $(v_1^n,v_2^n)$ is a sequence in $\mathbb{N}^{N-1} \times \mathbb{N}^{N-1}$ such that $P_n((v_1^n,v_2^n)) \rightarrow (u_1,u_2)$, then 
\begin{equation*}\label{microregenassumption} R^n(v_1^n,v_n^2,k) \circ P_n^{-1} \implies R(u_1,u_2,k)\end{equation*} in law in $\mathscr{M}((\mathbb{R}^+)^{N-1} \times (\mathbb{R}^+)^{N-1})$.
\end{enumerate}
With this in place, we define the processes $X_i^b$, $X_i^a$ stopping times $\tau_u^i$, $\tau_d^i$, $\tau^i$ and price sequence $m^i$ analogously to the previous section, with price jumps once again of size $\epsilon$. The underlying collections of exponential random variables which are used in the construction of the stopping times $\tau_u^i$ and $\tau_d^i$ are now denoted $Y_u^i$ and $Y_d^i$. Hence, $X_i$ is the sequence of static mesoscopic models with suitable mids, ($\tau_u^i, \tau_d^i)$ are stopping times determining the times in between price changes as well as the direction of these price changes, and $m^i$ is the sequence of mids. Our dynamic mesoscopic model is then given by $(\hat{X}^b(t),\hat{X}^a(t),m(t))$, where
\begin{equation*}\label{mesodynb}
\hat{X}^b(t):= \sum\limits_{i=1}^{\infty} X_i^b \left(t- \sum\limits_{j=1}^{i-1} \tau^j \right) \mathbbm{1}_{\left\{ \sum\limits_{j=1}^{i-1} \tau^j \leq t < \sum\limits_{j=1}^{i} \tau^j \right\}},
\end{equation*}
\begin{equation*}\label{mesodyna}
\hat{X}^a(t):= \sum\limits_{i=1}^{\infty} X_i^a \left(t- \sum\limits_{j=1}^{i-1} \tau^j \right) \mathbbm{1}_{\left\{ \sum\limits_{j=1}^{i-1} \tau^j \leq t < \sum\limits_{j=1}^{i} \tau^j \right\}},
\end{equation*}
\begin{equation*}
m(t):=\sum\limits_{i=1}^{\infty} m^i \mathbbm{1}_{\left\{ \sum\limits_{j=1}^{i-1} \tau^j \leq t < \sum\limits_{j=1}^{i} \tau^j \right\}}. 
\end{equation*}  

The following is our convergence result for approximating the dynamic mesoscopic model with our sequence of rescaled dynamic microscopic models. The mesoscopic model has the benefit of being more computationally efficient for large scale systems than the microscopic model, whilst maintaining a discrete spatial structure. 

\begin{thm}\label{microdynamic}
Suppose that $$\left(\frac{Z_{n,1}^a(0)}{\sqrt{n}},\frac{Z_{n,1}^b(0)}{\sqrt{n}}\right) \implies (X_1^a(0),X_1^b(0))$$ weakly in $\mathscr{M}\left((\mathbb{R}^+)^{N-1} \times (\mathbb{R}^+)^{N-1}\right)$. Let $(\hat{Z}_n^a(t),\hat{Z}_n^b(t),m_n(t))$ be dynamic microscopic models with initial data $$\left(\frac{Z_{n,1}^a(0)}{\sqrt{n}},\frac{Z_{n,1}^b(0)}{\sqrt{n}}, m^1 \right),$$ and let $(\hat{X}^a(t),\hat{X}^b(t),m(t))$ be the dynamic mesoscopic model with initial data $(X_1^a(0),X_1^b(0), m^1)$. Then $$\left(\frac{\hat{Z}_{n}^a(nt)}{\sqrt{n}},\frac{\hat{Z}_{n}^b(nt)}{\sqrt{n}}, m_n(nt) \right) \implies (\hat{X}^a(t),\hat{X}^b(t),m(t))$$ weakly in $\mathscr{M}\left(\mathbb{D}([0,\infty), \mathbb{R}^{N-1})\times \mathbb{D}([0,\infty), \mathbb{R}^{N-1})\times \mathbb{D}([0,\infty),\mathbb{R})\right)$.
\end{thm}
\begin{proof}
See Section A.2 in the appendix.
\end{proof}

\section{Reflected SPDEs: from mesoscopic to macroscopic models} 

In this section we rescale our previously obtained dynamic mesoscopic system and bridge the gap between mesoscopic and macroscopic models of limit order books. We show convergence to a dynamic macroscopic model limit by letting the tick size tend to zero and rescaling suitably. This connects our description of the order book to SPDE systems, which have been considered as order book models in the literature. See, for example, Zheng~\cite{Zheng}, M\"{u}ller~\cite{Mu} and \cite{KRM}. As with the dynamic result in the previous section, the proof of this relies heavily on convergence of the static models, which we present first. We are able to obtain convergence in the static setting by applying Theorem 2.1 in T. Zhang \cite{Z}.

\begin{rem}{\rm 
Note that, although we take the tick size to zero for the bid and ask sides of the book, the price process moves according to macroscopic price jumps. This both simplifies the analysis, avoiding the need to consider a model with a continuously moving boundary, and allows us to maintain a natural separation of time-scales for the order book evolution and price changes.} 
\end{rem}

\begin{rem}

{ \rm When simulating the model, any numerical scheme for the SPDE limit requires a projection which would essentially return us to an SDE framework. Nonetheless, it is interesting to note that the Poisson driven framework introduced in Section 2 has a natural SPDE analogue. The work in this section also completes our objective of connecting existing orderbook models in the literature at the particle, SDE and SPDE level, by noting that these correspond to one another if we rescale time and space appropriately.
}

\end{rem}
  
\subsection{SPDE limit in a static setting}

We begin this section by describing the sequence of rescaled static mesoscopic models which we will consider. For every $N \geq 1$ we let $X_N^b(t), X_N^a(t)$ satisfy the dynamics of the bid and ask sides of the static mesoscopic model on the price grid $\left\{0,1,2,...,N \right\}$ with mid $m \in \mathbb{R}$. These models are now indexed by $N$, which was suppressed in the previous section. We wish to emphasise this here since we will be taking $N$ to infinity. We will rescale space and time appropriately and map the coordinates of  $X_N^b(t)$ and $X_N^a(t)$ to equally spaced points on $[0,1]$. Define here the function $Q_N: \mathbb{R}^{N-1} \rightarrow C_0((0,1))$, such that
\begin{enumerate}[(i)]
\item For $i=1,2,...,N-1$, $Q_N(x)(\frac{i}{N})=\frac{x_i}{\sqrt{N}}$.
\item $Q_N(x)(0) = Q_N(x)(1) = 0$.
\item For $i=0,1,2,...,N-1$, $Q_N(x)$ is linear between the points $[\frac{i}{N},\frac{i+1}{N}]$.
\end{enumerate}
The aim will be to take the limit of $(Q_N(X_N^b(N^2t)),Q_N(X_N^a(N^2t)))$. We should have that the parameters for our SDE systems are consistent in some way. Given the rescaling, they are chosen such that for $k \in \{b,a\}$, $m \in \mathbb{R}$, $N \geq 1$, $i =0,1,...,N-1$ and $u \in (\mathbb{R}^+)^{N-1}$:

\begin{enumerate}[(i)]
\item $h^N_{k,m}(i,u):= N^{-\frac{3}{2}}h_{k,m}(\frac{i}{N},\frac{u}{\sqrt{N}})$,
\item $\sigma^N_{k,m}(i,u) := \sigma_{k,m}(\frac{i}{N}, \frac{u}{\sqrt{N}})$,
\end{enumerate}
where $h_{k,m}$, $\sigma_{k,m}$ are measurable maps from $[0,1] \times [0,\infty) \rightarrow \mathbb{R}$. We further assume the Lipschitz condition:
\begin{equation*}
|h_{k,m}(x,u)-h_{k,m}(y,v)|+ |\sigma_{k,m}(x,u)-\sigma_{k,m}(y,v)| \leq C(|x-y|+|u-v|).
\end{equation*}

We will show that our rescaled mesoscopic models converge to a reflected SPDE. Before doing so, we first give the definition of a solution to a reflected SPDE. This is the same as in T. Zhang \cite{Z}.

\begin{defn}\label{weaksolnref}

We say that the pair $(u, \eta)$ is a solution the SPDE with reflection
\begin{equation}\label{refSPDE}
\frac{\partial u}{\partial t}= \alpha \Delta u + h(x, u(t,x))+ \sigma(x,u(t,x)) \frac{\partial^2 W}{\partial x \partial t} + \eta(t,x)
\end{equation}
with Dirichlet conditions $u(t,0)=u(t,1)=0$ and initial data $u(0,x)= u_0 \in C_0((0,1))^+$ if
\begin{enumerate}[(i)]
\item u is a continuous adapted random field on $\mathbb{R}^+ \times [0,1]$ such that $u \geq 0$ almost surely.
\item $\eta$ is a random measure on $\mathbb{R}^+ \times (0,1)$ such that: 
\begin{enumerate}
\item For every $t \geq 0$, $\eta(\left\{t\right\} \times (0,1))=0$,
\item For every $t \geq 0$, $\int_0^t \int_0^1 x(1-x) \; \eta(\textrm{ds,dx}) < \infty$,
\item $\eta$ is adapted in the sense that for any measurable mapping $\psi$:
\begin{equation*}
\int_0^t \int_0^1 \psi(s,x) \; \eta(\textrm{ds,dx}) \; \;  \textrm{is }  \mathscr{F}_t-\textrm{measurable}.
\end{equation*}
\end{enumerate}
\item For every $t \geq 0$ and every $\phi \in C^2([0,1])$ with $\phi(0)=\phi(1)=0$,
\begin{equation*}
\begin{split}
\int_0^1 u(t,x) \phi(x) \textrm{d}x = &  \int_0^1 u(0,x)\phi(x) \textrm{d}x +\alpha \int_0^t \int_0^1 u(s,x)\phi^{\prime \prime}(x) \textrm{d}x \textrm{d}s \\ & + \int_0^t \int_0^1 h(x,u(s,x))\phi(x) \textrm{d}x \textrm{d}s + \int_0^t \int_0^1 \phi(x) \sigma(x,u(s,x)) W(\textrm{d}s,\textrm{d}x) \\ & + \int_0^t \int_0^1 \phi(x) \; \eta(\textrm{d}s,\textrm{d}x)
\end{split}
\end{equation*}
almost surely.
\item $\int_0^{\infty} \int_0^1 u(t,x) \;  \eta(\textrm{d}t,\textrm{d}x)=0$.
\end{enumerate}
\end{defn}
The intuition for this equation is that the reflection measure is analogous to the local time for a one dimensional diffusion. The solution follows the dynamics of a standard SPDE (without reflection), except when the solution meets the $x$-axis, where the profile is minimally pushed up by the reflection measure and kept positive.  \\

Existence of strong solutions of to equations of the form (\ref{refSPDE}), under our conditions on the coefficients, is proved by C.Donati-Martin and E. Pardoux in \cite{DP}. This means that, given a white noise process and the filtration that it generates, we can construct a solution which is adapted. Uniqueness was then proved by T. Xu and T. Zhang in \cite{XZ}.\\

We now pass to the macroscopic limit in the static setting. The following result can be shown by adapting Theorem 2.1 in T.Zhang \cite{Z}. We equip the space $C([0,\infty) ; C_0((0,1)))$ with the topology of uniform convergence on compact sets in the time variable, so $f_n \rightarrow f$ in $C([0,\infty) ; C_0((0,1)))$ if and only if $f_n \rightarrow f$ in $C([0,T] ; C_0((0,1)))$ for every $T>0$.  

\begin{thm}\label{mesostatic}
Suppose that $(Q_N(X_N^b(0)),Q_N(X_N^a(0))) \implies (v_0^b, v_0^a)$ in law in $C_0((0,1))^+ \times C_0((0,1))^+$. Then $(Q_N(X_N^b(N^2t)), Q_N(X_N^a(N^2t)) \implies (v^b,v^a)$ in law in $C([0,\infty); C_0((0,1)))$, where $(v^b,v^a)$ is the unique solution to the pair of reflected stochastic heat equations
\begin{equation*}
\frac{\partial v^b}{\partial t}= \alpha_b \Delta v^b + h_{b,m}(x,v^b)+ \sigma_{b,m}(x,v^b)\frac{\partial^2 W^b}{\partial x \partial t} + \eta^b(\textrm{d}t,\textrm{d}x),
\end{equation*}
\begin{equation*}
\frac{\partial v^a}{\partial t}= \alpha_a \Delta v^a + h_{a,m}(x,v^a)+ \sigma_{a,m}(x,v^a)\frac{\partial^2 W^a}{\partial x \partial t} + \eta^a(\textrm{d}t,\textrm{d}x),
\end{equation*}
with pinning conditions $v^b(t,0)=v^b(t,1)=v^a(t,0)=v^a(t,1)=0$ and initial data $v^b(0,x)=v_0^b(x)$ and $v^a(0,x)=v_0^a(x)$, where the space-time white noises $W^b$ and $W^a$ are independent. 
\end{thm}  

\begin{proof}
See Section A.3 in the appendix
\end{proof}

We call the limiting process $(v^b,v^a)$ our static macroscopic order book process with mid $m$.

\subsection{SPDE in a dynamic setting}

Let $(\hat{X}_N^b(t),\hat{X}_N^a(t),m_N(t))$ be our $N^{\textrm{th}}$ dynamic mesoscopic model, with the parameters defined as in the previous section. The price change rates are now denoted by $\theta^N_{b,m}$ and $\theta^N_{a,m}$, and the regenerative functions  by $R^N$. We aim to prove convergence to a dynamic macroscopic model by taking the limit of the sequence $(Q_N(\hat{X}^b_N(N^2t)),Q_N(\hat{X}^a_N(N^2t)), m_N(N^2t))$. 

We describe our dynamic macroscopic order book model. The principles are the same as for the dynamic microscopic and mesoscopic models, in that we assume the dynamics follow the static model in between price changes, with these occurring at state driven rates. As in previous sections, the stopping times determining the time in between the $(i-1)^{\textrm{th}}$ and $i^{\textrm{th}}$ price change, as well as the direction of the price change, are denoted $\tau^i_{u}$ and $\tau^i_{d}$, with $\tau^i := \tau^i_{u} \wedge \tau^i_{d}$. When an upward/downward price change is triggered, the price process $m(t)$ increases/decreases by $\epsilon$. The bid/ask profiles which give the evolution of the profile in between the $(i-1)^{\textrm{th}}$ and $i^{\textrm{th}}$ price change are denoted by $u_i^b$ and $u_i^a$ respectively. We introduce $\theta_{u,m}$ and $\theta_{d,m}$ for $m \in \mathbb{R}$, which are uniformly bounded over $m$ and continuous from $C_0((0,1)) \rightarrow \mathbb{R}_{>0}$, with there existing $c > 0$ such that $\theta_{u,m}, \theta_{d,m} \geq c$. These functions determine the rates of upward and downward price changes respectively as functions of the profile of the book when the mid is at $m$. We also introduce the regenerative function for our macroscopic model, $R$, which maps $C_0((0,1)) \times C_0((0,1)) \times \{u,d\}  \rightarrow \mathscr{M}(C_0((0,1) \times C_0((0,1)))$. This tells us the new profiles for the bid and ask sides of the book following a price change, as a function of the profiles at the time of the price change and the direction of the price change. The dynamic macroscopic model is then given by $(\hat{u}^b(t),\hat{u}^a(t),m(t))$, where
\begin{equation*}\label{mesodynbspde}
\hat{u}^b(t):= \sum\limits_{i=1}^{\infty} u_i^b \left(t- \sum\limits_{j=1}^{i-1} \tau^j \right) \mathbbm{1}_{\left\{ \sum\limits_{j=1}^{i-1} \tau^j \leq t < \sum\limits_{j=1}^{i} \tau^j \right\}},
\end{equation*}
\begin{equation*}\label{mesodynaspde}
\hat{u}^a(t):= \sum\limits_{i=1}^{\infty} u_i^a \left(t- \sum\limits_{j=1}^{i-1} \tau^j \right) \mathbbm{1}_{\left\{ \sum\limits_{j=1}^{i-1} \tau^j \leq t < \sum\limits_{j=1}^{i} \tau^j \right\}},
\end{equation*}
\begin{equation*}
m(t):=\sum\limits_{i=1}^{\infty} m^i \mathbbm{1}_{\left\{ \sum\limits_{j=1}^{i-1} \tau^j \leq t < \sum\limits_{j=1}^{i} \tau^j \right\}}. 
\end{equation*}  
For ease of notation, we introduce here the map $\tilde{Q}_N: \mathbb{R}^{N-1} \times \mathbb{R}^{N-1} \rightarrow C_0((0,1)) \times C_0((0,1))$, which simply applies $Q_N$ to each coordinate. We connect the dynamic aspects of the mesoscopic and macroscopic models by assuming that
\begin{enumerate}[(i)]
\item For $k \in \{b,a\}$, $m \in \mathbb{R}$, $(X^1,X^2) \in \mathbb{R}^{N-1} \times \mathbb{R}^{N-1}$ and $(u_1,u_2) \in C_0((0,1)) \times C_0((0,1))$ 
\begin{equation*}\label{mesothetaassumption} \left| N^2 \theta^N_{k,m}((X^1,X^2))- \theta_{k,m}((u_1,u_2)) \right| \leq r(\|\tilde{Q}_N((X_1,X_2))-(u_1,u_2)\|),\end{equation*}
where $\lim\limits_{x \rightarrow 0} r(x)=0$.
\item For $m \in \mathbb{R}$, if $(X^1,X^2) \in \mathbb{N}^{N-1} \times \mathbb{N}^{N-1}$ and $(u_1,u_2) \in C_0((0,1)) \times C_0((0,1))$ are such that \begin{equation*}\label{mesoregenassumption} \|\tilde{Q}_N((X_1,X_2)) - (u_1,u_2)\| \rightarrow 0,\end{equation*}
then, for $k \in \{b,a\}$, $$R^N(X^1,X^2,k)\circ \tilde{Q}_N^{-1} \implies R(u_1,u_2,k)$$ in law in $C_0((0,1)) \times C_0((0,1))$.
\end{enumerate}

\begin{rem}{\rm 
For $X \in \mathbb{N}^{N-1}$ and $u \in C_0((0,1))$, $$\|Q_N(X)-u\| \rightarrow 0$$ is equivalent to the condition that $$\sup\limits_{i=1,2,...,N-1} \left|X^i- u\left(i/N\right)\right| \rightarrow 0.$$}
\end{rem}

\begin{thm}\label{mesodynamic}
Suppose that $(Q_N(X_N^b(0)),Q_N(X_N^a(0))) \implies (u^b(0),u^a(0))$ in law in $C_0((0,1)) \times C_0((0,1)) $. Let $(\hat{X}_N^b(t), \hat{X}_N^a(t),m_N(t))$ be our dynamic mesoscopic model with initial data $$(X_N^b(0), X_N^a(0),m(0)),$$ and let $(\hat{u}^b(t), \hat{u}^a(t),m(t))$ be our dynamic macroscopic model with initial data $$(u^b(0),u^a(0),m(0)).$$ Then we have that 
\begin{equation*}
(Q_N(\hat{X}^b_N(N^2t)),Q_N(\hat{X}^a_N(N^2t)), m_N(N^2t)) \implies (\hat{u}^b(t), \hat{u}^a(t), m(t))
\end{equation*}
in law in the space $\mathbb{D}([0,\infty); C_0((0,1))) \times \mathbb{D}([0,\infty); C_0((0,1))) \times  \mathbb{D}\left([0,\infty); \mathbb{R} \right).$
\end{thm}

\section{Examples}

In this section we illustrate the flexibility of our set-up by discussing some ideas for different aspects of the model. 

\begin{eg}
\emph{Rate functions can be chosen such that $$\theta_{u,m}(u^1,u^2)= \gamma F\left(\int_0^{\epsilon} \left(u^1(x)-u^2(x)\right) \; \textrm{dx} \right) + \delta$$ and  $$\theta_{d,m}(u^1,u^2)=  \gamma F\left(\int_0^{\epsilon} \left(u^2(x)-u^1(x)\right) \; \textrm{dx} \right) + \delta$$
where $F$ is a non-negative continuous function, and $\gamma, \delta>0$. The rate at which price movements occur then has two components which are natural from a modelling standpoint. The first is a function of the local imbalance (the difference between the number of offers to buy and the number of offers to sell in a region close to the mid). The second is a fixed rate, intended to represent price movements due to exogenous factors. We will see in Section 5 how one might fit these parameters to data.}
\end{eg} 

\begin{eg}

\emph{We can incorporate large orders into our model. So far we have only directly considered small order sizes in our models, taking these wlog to be of size 1 in our microscopic models. We haven't, however, mentioned larger order sizes which would appear as ``jumps" in the macroscopic model profile in the limit. These can easily be incorporated by being assumed to appear (as one would expect) on a slower time scale to small orders. Large market orders which cause price changes are already accounted for in the existing set-up. We can also easily include extra stopping times $\tau_c^i$ into our models, which allow us to create a jump in the profile of the book without changing the price. This can model large cancellations or large limit orders for the book which do not trigger price changes. The rate at which these occur can then be given by some third rate function $\theta_{c,m}$, with the new profile again given by $R$ via the definition of $R(u^1,u^1,c)$, where we extend the definition of $R$ so that it now maps from $C_0((0,1)) \times C_0((0,1)) \times \{u,d,c\} \rightarrow \mathscr{M}(C_0((0,1)) \times C_0((0,1)))$. We note that in this model, the ``static" models in between stopping times do not refer to the evolution of the book in between price changes, but rather the evolution of the book in between both price changes and large orders.}

\end{eg}

\begin{eg}
\emph{As a particular case of the drift and volatility functions, $h_{k,m}$ and $\sigma_{k,m}$ we can take 
\begin{enumerate}
\item $h_{k,m}(x,u)= h_1(x)+ h_2(x)u$, and
\item $\sigma_{k,m}(x,u)= \sigma_1(x)+ \sigma_2(x)u$.
\end{enumerate}
The multiplicative terms $h_2(x)u$ and $\sigma_2(x)u$ can be thought of as self-exciting components for the order rates, whereby more orders on the book leads to faster trading. The remaining terms represent orders placed independently of the current order book profile.}
\end{eg}

\begin{eg}
\emph{We note that our regeneration functions, which determine the profiles of the two sides of the order book following price changes, allow us to choose profiles which both depend on the state of the book when the price change occurs, and are random. Therefore, natural deterministic choices, such as suitably removing orders from the previous profiles when the price changes, are permitted. We are also able to choose random profiles, such as sampling from invariant measures, or some combination of these two mechanisms.}
\end{eg}

\section{Numerical Investigation}

The aim of this final section is to demonstrate that even simple versions of the model can reproduce features of the price series and order books seen in
financial data. We describe the particular parameter choices for the model which will be used in this section, before giving an overview of the numerical 
scheme and introducing the LOBSTER dataset. Following this, we briefly explain the parameter estimation procedure used, before presenting numerical 
illustrations and results by plugging in the estimated parameters into the simulation algorithms. We would like to emphasise here that the intention 
behind this section is to demonstrate that sensible order book simulations can readily be obtained from the model - we do not claim that the numerical 
scheme implemented, or the methods used to fit the parameters, are optimal.

\subsection{The Model}

We will work with a special case of the models which fall within the framework described in the earlier sections. The two sides of the book, $u^b$ and $u^a$, will evolve in between price changes according to the reflected SPDEs
\begin{equation*}
\frac{\partial u^k}{\partial x}= \alpha \Delta u^k + f(x)+ \sigma(x) \frac{\partial^2 W^k}{\partial x \partial t} + \eta^k,
\end{equation*}
for $k \in  \{b, a \}$. The implicit assumption here is that the order arrival rates, which give rise to the drift and volatility terms $f$ and $\sigma$ in our SPDE limit, depend on the distance from the mid only, and there is no dependence on the number of orders which are currently on the book at that price. We have also imposed that the coefficients $f$ and $\sigma$ here do not depend on $k$. This is due to the symmetry in the order arrival rates for the bid and ask sides of the book, which can be seen in the data.  

We will now describe the form of the rate functions which will be used in this section, $\theta_u(u)$ and $\theta_d(u)$, which determine the rates at which the price moves up and down respectively. In a particular case of Example 1 of Section 4, the rate at time $t$ for an upward price jump will be given by 
\begin{equation}\label{rate1}
\theta_u(u^b(t,\cdot),u^a(t,\cdot))=\gamma \max \left(\int_0^{\epsilon} \left( u^b(t,x) - u^a(t,x) \right) \textrm{d}x, 0
 \right) + \delta
\end{equation}
and the rate of a downward price movement will similarly be given by
\begin{equation}\label{rate2}
\theta_d(u^b(t,\cdot),u^a(t,\cdot))=\gamma \max \left(\int_0^{\epsilon} \left( u^a(t,x) - u^b(t,x) \right) \textrm{d}x, 0
 \right) + \delta.
\end{equation}
The first terms in these rates represents the contribution of order imbalance close to the mid as a driving factor for price movement, whilst the second terms here represent price movements due to additional exogenous factors. 

Our regeneration functions, which determine the profiles of the two sides of the book following price changes, will simply shift the profiles of the two sides of the book by the size of a price jump, in the relevant direction. That is, $R$ is given by 
\begin{equation}\label{reinitialise} \begin{split}
R\left(u^b, u^a, u \right)=
\begin{cases}
(0, u^a(x+\epsilon))  \hspace*{4mm}\text{for}\hspace*{2mm}x\in[0,\epsilon],\\ 
(u^b(x-\epsilon), u^a(x+\epsilon))  \hspace*{4mm}\text{for}\hspace*{2mm}x\in[\epsilon,1-\epsilon],\\
(u^b(x-\epsilon), 0)  \hspace*{4mm}\text{for}\hspace*{2mm}x\in[1-\epsilon,1],\\
\end{cases} \\
R\left(u^b ,u^a, d \right)=
\begin{cases}
(u^b(x+\epsilon),0)  \hspace*{4mm}\text{for}\hspace*{2mm}x\in[0,\epsilon],\\ 
(u^b(x+\epsilon), u^a(x-\epsilon))  \hspace*{4mm}\text{for}\hspace*{2mm}x\in[\epsilon,1-\epsilon],\\
(0,u^a(x-\epsilon))  \hspace*{4mm}\text{for}\hspace*{2mm}x\in[1-\epsilon,1],\\
\end{cases}
\end{split}
\end{equation}

We note here that, strictly speaking, some of these new profiles do not satisfy the Dirichlet conditions imposed earlier in the analysis. However, the equations can be shown to have solutions when started from a general positive continuous initial profile, with the solution in $C_0((0,1))^+$ at all positive times. In addition, the regeneration function as described here has a natural interpretation simply as the best bid/ask queue having been depleted to trigger the price change. 

\subsection{The Numerical Scheme}

We use a forward time-stepping scheme in order to simulate our equations in between price movements. The equation is discretised into $M$ time steps and $N$ space steps, and we define $t_i:= iT/M$, $x_i:= i/N$. We denote the simulated height of the bid/ask sides of the book at time $t_j$ and position $x_i$ by $u^b(t_j,x_i)$ and $u^a(t_j,x_i)$ respectively. The simulated price process at time $t_j$ is denoted by $p(t_j)$. Given our simulated solution up to the $j^{\textrm{th}}$ time step, we begin our approximation of the following time-step by first determining whether there should be a price movement at this time-step. Note that we should take the price jump size, $\epsilon$, to be equal to a multiple of the space step in that $\epsilon=k/N$ for an integer $k$. In our simulations, we will take $k=1$. We start by defining $\pi^+[t_j]$ and $\pi^-[t_j]$
which give the probabilities of upward of downward price jumps in that time-step. These are approximations of the probabilities of price jumps in that time period given by the rates $\theta_u(u^b(t,\cdot),u^a(t,\cdot))$ and $\theta_d(u^b(t,\cdot),u^a(t,\cdot))$ as in (\ref{rate1}) and (\ref{rate2}). Writing these explicitly, we have that 
\begin{equation*}
\pi^+[t_j]= \left( \max(\frac{\gamma}{2N} \left(u^b(t_j,x_1)- u^a(t_j,x_1) \right),0) +\delta \right) \times T/M.
\end{equation*}
and 
\begin{equation*}
\pi^-[t_j]= \left(\max(\frac{\gamma}{2N} \left(u^a(t_j,x_1)- u^b(t_j,x_1) \right), 0) + \delta \right) \times T/M.
\end{equation*}
We simulate a uniform random variable on $[0,1]$, 
which we denote by $Y[t_j]$. If $Y[t_j] < \pi^+[t_j]$ we take this as an indication that the price has moved up in that time-step. 
Similarly, if $\pi^+[t_j] \leq Y[t_j]<\pi^+[t_j]+\pi^-[t_j]$, 
we take this as an indication that there has been a downward price change at this time-step. If we are in neither of these cases, the simulated order book does not change price during this time-step. If there has been a price increase, we update our price process by setting $p[t_{j+1}]=p[t_j] + \epsilon$, and adjust the profiles $u^b$ and $u^a$ by discrete approximations to (\ref{reinitialise}). That is, we simply set
\begin{equation*}
u^b(t_j,x_0)= u^a(t_j,x_{N})=0,
\end{equation*}
and define 
\begin{equation*}
u^b(t_j,x_i)=u^b(t_j,x_{i-1})
\end{equation*}
for $i \in \{1,..,N\}$, and 
\begin{equation*}
u^a(t_j,x_i)=u^a(t_j,x_{i+1})
\end{equation*}
for $i \in \{0,...,N-1 \}$. We analogously update $u^b$, $u^a$ and $p$ in the event of a downward price movement. If there is no price change at this time-step, we do not update $u^b$, $u^a$ at this point, and set $p(t_{j+1})=p(t_j)$. At the end of this price updating procedure, we then simulate the profiles of $u^b$ and $u^a$ at the next time-step by setting, for $i \in \{1, 2,...,N-1 \}$
\begin{equation*}\label{num1}\begin{split}
u^b(t_{j+1},x_i):= & \max \left\{ u^b(t_j,x_i)+ \frac{T N^2}{M}\left( u^b(t_j,x_{i+1})+ u^b(t_j,x_{i-1})-2u^b(t_j,x_i) \right) \right.  \\ & \left. + \frac{T}{M}f(x_i) + \frac{\sqrt{TN}}{\sqrt{M}}\sigma(x_i)Z_{i,j}^b,0   \right\},
\end{split}
\end{equation*}
and similarly 
\begin{equation*}\label{num2}\begin{split}
u^a(t_{j+1},x_i):= & \max \left\{ u^a(t_j,x_i)+ \frac{TN^2}{M}\left( u^a(t_j,x_{i+1})+ u^a(t_j,x_{i-1})-2u^a(t_j,x_i) \right) \right.  \\ & \left. + \frac{T}{M}f(x_i) + \frac{\sqrt{TN}}{\sqrt{M}}\sigma(x_i)Z_{i,j}^a ,0  \right\}.
\end{split}
\end{equation*}
The $Z_{i,j}^k$ here are simply simulated unit normal random variables, and appear due to the discretisation of the space-time white noise component of our equations. We take the maximum with zero at each time-step in order to capture the influence of the reflection measures in the equations. This step completes our forward time-step, and we return the profiles $u^b(t_{j+1},\cdot)$, $u^a(t_{j+1},\cdot)$ and the price $p(t_{j+1})$ as our simulated order book at time $t_{j+1}$.

\begin{rem}

{\rm By simulating the problem in this way, we can also think of our numerical scheme as a simulation of a mesoscopic model consisting of 50 coupled SDEs. Our work from Section 3 shows that, for sufficiently many queues, the SDE model is essentially an SPDE model. We choose to present our numerics here as a discretisation of the macroscopic SPDE model.}

\end{rem}


\subsection{Dataset Description}

The data we have at our disposal originates from the LOBSTER (Limit Order Book System, The Efficient Reconstructor) database project initiated by the Humboldt University of Berlin, which gives access to reconstructed limit order book data for all NASDAQ traded stocks between June 2007 up to the present day. For each trading day of a given ticker, LOBSTER generates two distinct files. On the one hand, a \textit{message} file, which lists indicators for the different kinds of events which cause an update of the book (limit order arrivals and cancellations, executions or market orders, trading halts) within a prespecified price range. On the other hand, an \textit{order book} file, which displays the evolution of the book up to a chosen number of \textit{occupied} price levels (which can go up to 200, depending on the selected ticker). Order book events are timestamped according to seconds after midnight, and the decimal precision available ranges from milliseconds to nanoseconds. Our sample consists of data from the SPDR Trust Series I, and covers the 50 best levels on each side of the book on June 21 2012 between 11:00:00.000 and 12:00:00.000 EST, with tick size size being $1$ cent for the dataset.

\subsection{Parameter Fitting}

In this section, we aim to obtain input parameters for our model based on the data described in the previous section. We would like to emphasise here that we do not claim that our fitting methods in this section are particularly robust - we simply aim to fit the parameters in a reasonable way so that we can demonstrate that our model can provide sensible simulations of the evolution of the order book. In addition, we do not fit the smoothing parameter, $\alpha$ and instead choose this value such that the profiles obtained are on the correct scale. In order to account for the reflection measure, we will only use the data of limit orders placed in price levels which are currently occupied. This is done since we should only take into consideration the order arrivals when the order volumes are away from zero, so as not to bias our drift estimates with components which could be attributed to the reflection measure. 

Throughout this section, we will denote by $X^b_{j,i}$ the order volume in the $i^{\textrm{th}}$ price point below the best bid at the $j^{\textrm{th}}$ time-step for our dataset. We similarly denote by $X^a_{j,i}$ the order volume in the $i^{\textrm{th}}$ price point above the best ask at the $j^{\textrm{th}}$ time-step for our dataset. 

The scaling used will measure order volumes in units of $10^4$ orders, and the $50$ queues of the order book will map to positions $i/51$. This simply results in the SPDEs for each side of the book each being on a spatial interval of length 1. Since each queue of our dataset represents a price change of size $1$ cent, a spatial interval of size $1/51$ corresponds to $1$ cent under our scaling. Time units for the SPDEs will be measured in minutes. In this scale, our choice of smoothing parameter, $\alpha$ will be given by $\alpha=0.01$.

\subsection{Volatility and Drift Estimates}

We begin this section by fitting the volatility and drift parameters for our SPDE. Let $\hat{\sigma}(i/51)$ denote our estimate for $\sigma(i/51)$, the volatility at 
spatial position $i/51$, corresponding to $i$ price points away from the best bid/ask. We calculate $\hat{\sigma}$ by equating
\begin{equation*}
51 \times 60 \times \hat{\sigma}^2(i/51) = \frac{1}{2} \sum \left[ \textrm{Order Size} \times 10^{-4} \right]^2,
\end{equation*}
where the sum is taken over all orders on the bid/ask side of the book of all types which are $i$ price points away from the best bid/ask price respectively, with the omission of limit orders placed in unoccupied queues. This simply matches the quadratic variation of different spatial intervals in the model with the corresponding value from the dataset. The factors of $51$ and $3600$ appear here due to our space and time scaling respectively, whilst the $1/2$ is 
present simply as we are summing data from both sides of the book, rather than just one side. The factor of $10^{-4}$ appears as we are working with 
units of $10^4$ orders. 

We now fit the drift term. Let $d^b_i$ be the total number of bid limit orders placed in the hour long period in queues which are $i$ price points away from the 
best bid, once again disregarding those orders placed in unoccupied queues. We similarly define $d^a_i$, with $c^b_i$ and $c^a_i$ the corresponding values for market/cancellation orders at the different relative price points for the bid and ask sides of the book respectively. The net order flow over the hour long period at a price point which is $i$ ticks away from the best bid/ask queue is then obtained from the dataset as  
\begin{equation*}
\frac{1}{2}\left( d^a_i+d^b_i -c^a_i -c^b_i \right) \times 10^{-4}.
\end{equation*}
Denote by $\tilde{f}(i/51)$ our estimate for the drift term at spatial position $i$ and recall that we are working in time units of minutes. We estimate the drift parameter by equating
\begin{equation}\label{driftfit}
60 \times (\hat{f}(i/51) + 0.01\hat{\Delta}_iu)  = \frac{1}{2}\left( d^a_i+d^b_i -c^a_i -c^b_i \right) \times 10^{-4}.
\end{equation}
The term $0.01\hat{\Delta}_iu$ here represents the contribution of the laplacian term in our equation to the net order flow per minute at the $i^{\textrm{th}}$ level. $\hat{\Delta}_iu$ is obtained by calculating the average \emph{Dirichlet} laplacian at the bid/ask sides of the book at price points which are $i$ ticks away from the best bid/ask prices. By Dirichlet here, we mean that the laplacian is calculated at the first and last queues using the convention that the $0^{\textrm{th}}$ and $51^{\textrm{th}}$ queues are empty.

Figure 1 displays the estimated drift and volatility functions obtained from the data by the techniques described above. We note that the volatility is at its 
largest close to the mid, representing that there was significantly more activity at these price points over the trading period, as one would expect.

\begin{figure}[H]
        \centering
        \begin{subfigure}[b]{0.45\textwidth}
            \centering
            \includegraphics[width=\textwidth]{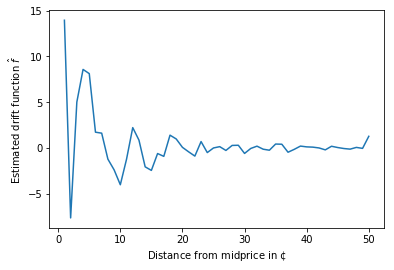}
            \caption[]%
            {{}}    
        \end{subfigure}
        \quad
        \begin{subfigure}[b]{0.45\textwidth}   
            \centering 
            \includegraphics[width=\textwidth]{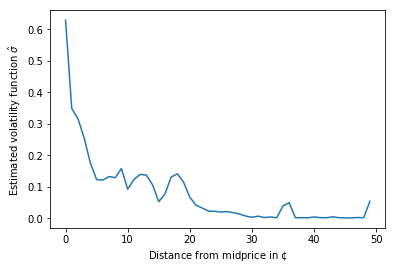}
            \caption[]%
            {{}}    
        \end{subfigure}
        \caption[]
        {\footnotesize Estimated Drift (a), and volatility (b).} 
\end{figure}

\subsubsection{Estimates for Rates of Price Changes}

It is only left to produce estimates for $\gamma$, the rate at which the model changes price due to imbalance of the bid and ask queues, and $\delta$, the rate at which the model changes price due to exogenous factors. Let $P(t)$ denote the price process of our dataset, measured in cents (recall that the tick size for the data is one cent). Our estimate of $\gamma$, $\hat{\gamma}$, is chosen such that it satisfies the equation
\begin{equation}\label{P_eqn}
P(1)-P(0)= \hat{\gamma} \times 60 \times I.
\end{equation}
The value $I$ here is the average local imbalance of the data over the entire period, given by
\begin{equation*}
I= \frac{1}{J} \sum\limits_{j=1}^J \left[ \int_0^{1/51} \left(\tilde{X}^b_j(x)- \tilde{X}^a_j(x) \right) \; \textrm{d}x \right],
\end{equation*}
where $J$ is the number of timesteps for our dataset, and $\tilde{X}^b_j(x)$ and $\tilde{X}^a_j(x)$ are obtained from $X^b_{j,i}$ and $X^a_{j,i}$ by setting $\tilde{X}^b_j(i/51)=10^{-4}X^b_{j,i}$ and $\tilde{X}^a_{j}(i/51)=10^{-4}X^a_{j,i}$, and then linearly interpolating in between these points. Writing $I$ in terms of $X^b$ and $X^a$, we have
\begin{equation*}
I= \frac{1}{2 \times 51 \times 10^4 \times J} \sum\limits_{j=1}^J \left( X^b_{j,1}-X^a_{j,1} \right).
\end{equation*}
The factor of $60$ appears in (\ref{P_eqn}) since we are working in units of minutes, and we have used an hour of data.
The rate for price movements in either direction due to exogenous factors, $\hat{\delta}$, is then fitted so that the quadratic variation of the price process 
from the dataset matches with the expected quadratic variation had the price moved due to our price changing mechanism with parameters $\hat{\gamma}$ 
and $\hat{\delta}$. The expected number of price changes due to the local imbalance component is given by $60 \hat{\gamma} \tilde{I}$, where
\begin{equation*}
\tilde{I}=\frac{1}{2 \times 51 \times 10^{4} \times J} \sum\limits_{j=1}^J \left| X^b_{j,1}-X^a_{j,1} \right|.
\end{equation*}
For a given $\delta$ we would expect an extra $2 \times 60 \times \delta$ price changes over the time period. Our parameter $\hat{\delta}$ is therefore chosen such that it solves
\begin{equation*}
120 \hat{\delta}= \left[ \sum\limits_{j=1}^J \left(P(j/J)-P((j-1)/J)\right)^2 \right] - 60\hat{\gamma}\tilde{I}.
\end{equation*}
Implementing the procedures described, we obtain the following values for $\hat{\delta}$ and $\hat{\gamma}$. 
\begin{table}[H]
\centering 
\begin{tabular}{ccccccc}
\hline\hline
$\hat{\gamma}$ & $\hat{\delta}$\\ \hline
2720 & 12.76 \\
\hline
\end{tabular}
\label{tab:hresult}
\caption{Estimated Price Change Rates.} 
\end{table}

\subsection{Simulation Results}

In this section, we will present the results of our numerical simulation of the order book under our model. In addition to the parameters which were fitted in the previous section, we fix the following additional parameters for the simulation.

\begin{table}[H]
\centering 
\begin{tabular}{ccccccc}
\hline\hline
$\alpha$ & $T$ & $N$ & $M$  \\ \hline
0.01 & 60 & 50 & 1500000 \\
\hline
\end{tabular}
\label{tab:hresult2}
\caption{Additional Model Parameters.} 
\end{table}

Recall that $\alpha$ refers to the smoothing parameter, $T$ the time period in minutes for which we run the simulation and $N$, $M$ the number of space and time steps used respectively.

We will now present graphs illustrating the outcome of our simulations. In order to emphasise the strength of the fit, we present graphs of our simulations next to the corresponding graphs from the dataset.

\begin{figure}[H]
        \centering
        \begin{subfigure}[b]{0.45\textwidth}
            \centering
            \includegraphics[width=\textwidth]{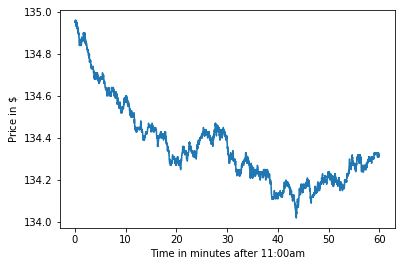}
            \caption[]%
            {{}}    
        \end{subfigure}
        \quad
        \begin{subfigure}[b]{0.45\textwidth}   
            \centering 
            \includegraphics[width=\textwidth]{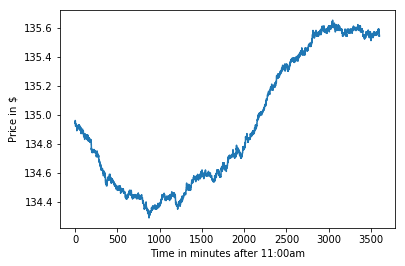}
            \caption[]%
            {{}}    
        \end{subfigure}
        \caption[]
        {\footnotesize Price process of SPDR1  (a), simulated price process from fitted model (b).} 
\end{figure}
\begin{figure}[H]
        \centering
        \begin{subfigure}[b]{0.45\textwidth}
            \centering
            \includegraphics[width=\textwidth]{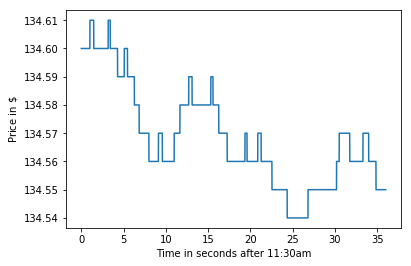}
            \caption[]%
            {{}}    
        \end{subfigure}
        \quad
        \begin{subfigure}[b]{0.45\textwidth}   
            \centering 
            \includegraphics[width=\textwidth]{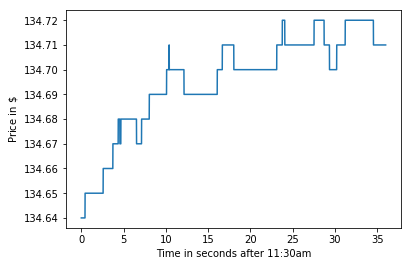}
            \caption[]%
            {{}}    
        \end{subfigure}
        \caption[]
        {\footnotesize Price process of SPDR1  (a), simulated price process from fitted model (b).} 
\end{figure}
We note that the quadratic variation of the price process from the dataset over the hour long time period is $0.2417$ (where the y-axis is measured in dollars), and the quadratic variation of the simulated path is $0.2495$. When fitting, we attributed $0.1531$ of the quadratic variation of the price from our dataset to exogenous price movements, with the remaining considered to be due to local imbalance. Based on this fitting, $0.1531$ of the quadratic variation for our simulated price process can also be attributed to exogenous movements, and once again the remainder here is due to local 
imbalance. That is, local imbalance contributed to $0.0964$ of the quadratic variation of the simulated price process, compared to $0.0886$ for the dataset. Therefore, price movements due to local imbalance were produced at a realistic rate by the simulation. This is compatible with the fact that the average order book profile over the hour long period appears to match well with that which was observed in the original data, 
which would mean that in particular we might expect the local imbalance for the real data and the simulated process to be of similar magnitude 
on average. The following figures show the average profiles of the ask side of the book for the real and simulated processes, together with snapshots 
of these profiles.

\begin{figure}[H]
        \centering
        \begin{subfigure}[b]{0.45\textwidth}
            \centering
            \includegraphics[width=\textwidth]{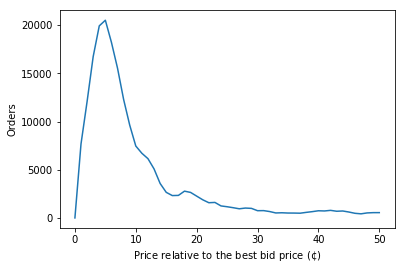}
            \caption[]%
            {{}}    
        \end{subfigure}
        \quad
        \begin{subfigure}[b]{0.45\textwidth}   
            \centering 
            \includegraphics[width=\textwidth]{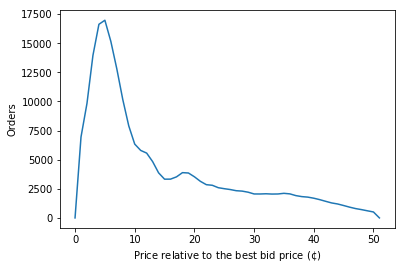}
            \caption[]%
            {{}}    
        \end{subfigure}
        \caption[]
        {\footnotesize Time averaged order book profiles for the ask side of SPDR1 (a) and our simulation (b).} 
\end{figure}
\begin{figure}[H]\label{fig:book-profile}
        \centering
        \begin{subfigure}[b]{0.45\textwidth}
            \centering
            \includegraphics[width=\textwidth]{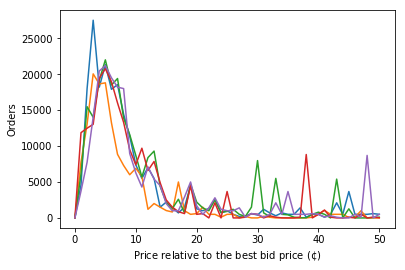}
            \caption[]%
            {{}}    
        \end{subfigure}
        \quad
        \begin{subfigure}[b]{0.45\textwidth}   
            \centering 
            \includegraphics[width=\textwidth]{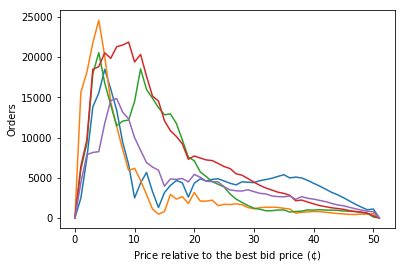}
            \caption[]%
            {{}}    
        \end{subfigure}
        \caption[]
        {\footnotesize Static snapshots of order book profiles at 12 minute intervals for the ask side of SPDR1 (a) and our simulation (b).} 
\end{figure}

We note that the average profile from our simulation is on the correct scale, and, like the average profile from the data, has most mass concentrated close to the mid. The snapshots from the simulation also demonstrate this. 

\appendix

\section{Appendix}

\subsection{Proof of Theorem \ref{microstatic} }

\begin{proof}
Since the particle/SDE systems decouple into two independent systems, it is sufficient to prove convergence of $\tilde{Z}^b_n$ to $X^b$ in $\mathscr{M}(\mathbb{D}([0,\infty); \mathbb{R}^{N-1}))$. It follows from the dynamics of $Z^b_n$ that the rescaled process, $\tilde{Z}_n^b(t)$, has dynamics given by:
\begin{enumerate}[(i)]
\item For $i \in \left\{1,...,N-1 \right\}$, $\tilde{Z}_n^b \rightarrow \tilde{Z}_n^b + \frac{e_i}{\sqrt{n}}$ at exponential rate 
\begin{equation*}\begin{split}
\frac{n}{2} & \sigma_{b,m}^2 (i, \tilde{Z}_n^{b,i})\left( 1+ \mathbbm{1}_{\left\{\tilde{Z}_n^{b,i}=0\right\}} \right) + \sqrt{n}f_{b,m}\left(i, \tilde{Z}_n^{b,i} \right).
\end{split}
\end{equation*}
\item For $i \in \left\{1,...,N-1 \right\}$, $\tilde{Z}_n^b \rightarrow \tilde{Z}_n^b -\frac{e_i}{\sqrt{n}}$ at exponential rate 
\begin{equation*}
\begin{split}
\frac{n}{2} &  \sigma_{b,m}^2 (i, \tilde{Z}_n^{b,i}) \mathbbm{1}_{\left\{\tilde{Z}_n^{b,i} \geq \frac{1}{\sqrt{n}} \right\}} + \sqrt{n}g_{b,m}\left(i, \tilde{Z}_n^{b,i} \right)\mathbbm{1}_{\left\{ \tilde{Z}_n^{b,i} \geq \frac{1}{\sqrt{n}} \right\}}.
\end{split}
\end{equation*}
\item  For $i \in \left\{1,...,N-1 \right\}$, $\tilde{Z}_n^b \rightarrow \tilde{Z}_n^b+\frac{e_{i-1}}{\sqrt{n}}-  \frac{e_i}{\sqrt{n}}$ at exponential rate 
\begin{equation*}
\alpha_{b}\sqrt{n} \tilde{Z}_n^{b,i}
\end{equation*}
\item  For $i \in \left\{1,...,N-1 \right\}$, $\tilde{Z}_n^b \rightarrow \tilde{Z}_n^b+\frac{e_{i+1}}{\sqrt{n}}-  \frac{e_i}{\sqrt{n}}$ at exponential rate 
\begin{equation*}
\alpha_{b}\sqrt{n} \tilde{Z}_n^{b,i}
\end{equation*}
\end{enumerate}
We note that the indicator functions above will ensure that our processes converge to \emph{reflected} SDEs. In order to prove convergence of these processes, we will argue that their infinitesimal generators converge (see Corollary 4.8.7, \cite{EK}). We start by calculating the infinitesimal generators for the processes $\tilde{Z}_n^b$. Define: 
\begin{equation*}\begin{split}
& \Delta^l_{n,k} F(y):= \sqrt{n}\left[F(y)-F\left(y-\frac{e_k}{\sqrt{n}}\right)\right], \\
& \Delta^r_{n,k} F(y):= \sqrt{n}\left[F\left(y+\frac{e_k}{\sqrt{n}}\right)-F\left(y\right)\right], \\
& \Delta^{2}_{n,k}:= n\left[F\left(y+\frac{e_k}{\sqrt{n}}\right)+F\left(y-\frac{e_k}{\sqrt{n}}\right)-2F(y)\right].
\end{split}
\end{equation*}  
For ease of notation, we use the convention that terms are zero whenever they refer to the $0^{\textrm{th}}$ or $N^{\textrm{th}}$ queues. Then we have that, for $F \in C_0([0, \infty)^{N-1}))$, the continuous functions on $[0,\infty)^{N-1}$ which vanish at $\infty$, 
\begin{equation*}\begin{split}
\frac{1}{t}\left( \mathbb{E}\left[F(\tilde{Z}_n(t))|\tilde{Z}_n(0)=y \right] - F(y) \right)= & \sum\limits_{k=1}^{N-1} \Delta_{n,k}^r F(y) \left(\frac{\sqrt{n}}{2}\sigma^2_{b,m} \left(k,y_k\right)\left(1+\mathbbm{1}_{\left\{y_k=0\right\}}\right)\right)\\ &+ \sum\limits_{k=1}^{N-1} \Delta_{n,k}^rF(y)f_{b,m} \left(k,y_k\right) \\ &+\sum\limits_{k=1}^{N-1}  \left(\Delta_{n,k-1}^r F(y-\frac{e_k}{\sqrt{n}})-\Delta_{n,k}^lF(y) \right)y_k  \\ &+ \sum\limits_{k=1}^{N-1} \Delta_{n,k}^l F(y) \left(-\frac{\sqrt{n}}{2}\sigma^2_{b,m} \left(k,y_k\right)\mathbbm{1}_{\left\{y_k \geq \frac{1}{\sqrt{n}}\right\}}\right)\\ &+ \sum\limits_{k=1}^{N-1} \Delta_{n,k}^lF(y)\left(-g_{b,m}\left(k,y_k\right)\mathbbm{1}_{\left\{y_k \geq \frac{1}{\sqrt{n}}\right\}} \right) \\ &+\sum\limits_{k=1}^{N-1} \left(\Delta_{n,k+1}^r F(y- \frac{e_k}{\sqrt{n}})-\Delta_{n,k}^lF(y) \right)y_k  \\ &+ R_{y,t},  
\end{split}
\end{equation*}
where $R_{y,t} \rightarrow 0$ as $t \rightarrow 0$ uniformly in compact sets for $y$. Note that this control on $R_{y,t}$ is a consequence of the local boundedness of $\sigma$, $f$ and $g$, since these conditions ensure that the jump rate of the process after the first jump from $y$ can be bounded. If we further assumed that $F$ were smooth with compact support, so $F \in C_c^{\infty}([0,\infty)^{N-1})$, we obtain that the remainder term $R_{y,t}$ converges uniformly to zero for $y \in (\frac{1}{\sqrt{n}}\mathbb{N})^{N-1}$. Therefore, writing $A_n$ for the generator of $\tilde{Z}_n$ and rearranging gives that for $F \in C_c^{\infty}([0,\infty)^{N-1})$,
\begin{equation*}\begin{split}
A_nF(y)=& \sum\limits_{k=1}^{N-1} \frac{1}{2} \Delta^2_{n,k}F(y) \mathbbm{1}_{\left\{y_k \geq \frac{1}{\sqrt{n}}\right\}} \sigma^2_{b,m}(k,y_k) + \sum\limits_{k=1}^{N-1} \Delta_{n,k}^rF(y) \sqrt{n}\sigma^2_{b,m}(k,y_k) \mathbbm{1}_{\left\{y_k=0\right\}} 
\\ &+ \sum\limits_{k=1}^{N-1} \Delta_{n,k}^rF(y)f_{b,m}(k,y_k) - \sum\limits_{k=1}^{N-1} \Delta_{n,k}^lF(y)g_{b,m}(k,y_k)\mathbbm{1}_{\left\{y_k \geq \frac{1}{\sqrt{n}}\right\}} 
\\ & + \sum\limits_{k=1}^{N-1} \left( \Delta_{n,k-1}^r F(y - \frac{e_k}{\sqrt{n}}) + \Delta_{n,k+1}^r F(y - \frac{e_k}{\sqrt{n}}) -2\Delta_{n,k}^l F(y) \right)y_k.
\end{split}
\end{equation*}
Recall that our candidate limiting process, $X^b$, satisfies the system of reflected SDEs 
\begin{equation*}\label{reflectedsde}
\textrm{d}X_t^{b,i}= \alpha_b(X_t^{b,i+1}+X_t^{b,i-1}-2X_t^{b,i}) \textrm{dt}+h_{b,m}(i,X_t^{b,i})\textrm{dt}+ \sigma_{b,m} (i,X_t^{b,i}) \textrm{d}W_t^{b,i} + \textrm{d}\eta_t^{b,i},
\end{equation*}
for $i=1,2,...,N-1$, with $X^0=X^N=0$ and $h_{b,m}:=f_{b,m}-g_{b,m}$. Note that by Theorem 4.1 in \cite{T}, we have existence of a strong solution and pathwise uniqueness for this system of SDEs. Inspection of the proof also reveals that the solution here, where the diffusion and drift coefficients do not have any explicit time-dependence, is continuous. We can calculate the corresponding generator to be 
\begin{equation*}
AF(x)= \frac{1}{2} \sum\limits_{k=1}^{N-1} \sigma_{b,m}^2(k,x_k) \frac{\partial^2F}{\partial x_k^2} \;  + \; \sum\limits_{k=1}^{N-1} \left[h_{b,m}(k,x_k)+ \alpha_b(x_{k+1}+x_{k-1}-2x_k) \right] \frac{\partial F}{\partial x_k},
\end{equation*} 
acting on the domain 
\begin{equation*}
\mathscr{D}(A)= \left\{ F \in C_0^2([0,\infty)^{N-1}) \; \; \textrm{s.t.} \; \;  \forall k, \; \;  \frac{\partial F}{\partial x_k}\bigg|_{x_k=0}=0 \right\}.
\end{equation*}
This has a core given by 
\begin{equation*}
C(A)= \left\{ F \in C_c^{\infty}([0,\infty)^{N-1}) \; \; \textrm{s.t.} \; \;  \forall k, \; \;  \frac{\partial F}{\partial x_k} \bigg|_{x_k=0}=0 \right\}.
\end{equation*}
In our setting it is enough to prove (see Corollary 4.8.7 in \cite{EK}) that for all $F \in C(A)$,
\begin{equation*}\label{gen}
\sup\limits_{y \in \frac{1}{\sqrt{n}} \mathbb{N}^{N-1}}|A_nF(y)-AF(y)| \rightarrow 0.
\end{equation*}
First suppose that $y_k \geq \frac{1}{\sqrt{n}}$. Then, by Taylor's Theorem we have that 
\begin{equation*}
\frac{1}{2} \sigma_{b,m}^2(k,y_k) \left|\Delta_{n,k}^2F(y)- \frac{\partial^2F}{\partial y_k^2}(y)\right| \leq \frac{1}{6\sqrt{n}} \left\|\frac{\partial^3F}{\partial y_k^3}\right\|_{\infty} \|\sigma_{b,m}^2\|_{F, \infty},
\end{equation*}
where $\|\sigma_{b,m}^2\|_{F, \infty}$ is the supremum of $\sigma_{b,m}^2$ over the support of F.
Similarly 
\begin{equation*}
\left|\Delta_{n,k}^rF(y)- \frac{\partial F}{\partial y_k}(y)\right|f_{b,m}(k,y_k) \leq \frac{1}{2\sqrt{n}} \left\| \frac{\partial^2 F}{\partial y_k^2} \right\|_{\infty} \|f_{b,m}\|_{F, \infty},
\end{equation*}
and
\begin{equation*}
\left|\Delta_{n,k}^lF(y)- \frac{\partial F}{\partial y_k}(y)\right|g_{b,m}(k,y_k) \leq \frac{1}{2\sqrt{n}} \left\| \frac{\partial^2 F}{\partial y_k^2} \right\|_{\infty} \|g_{b,m}\|_{F, \infty}.
\end{equation*}
Now suppose that $y_k =0$. Then, again by Taylor's theorem (extending F by reflection about the axes and using that $\frac{\partial F}{\partial y_k}$ is zero when $y_k =0$) we have that 
\begin{equation*}
\begin{split}
& \sigma_{b,m}^2(k,y_k)  \left| \sqrt{n}\Delta^r_{n,k}F(y)-\frac{1}{2} \frac{\partial^2F}{\partial y_k^2}(y) \right| \leq \frac{1}{6\sqrt{n}}\left\|\frac{\partial^3F}{\partial y_k^3}\right\|_{\infty}\|\sigma_{b,m}^2\|_{F, \infty}, 
\\ & f_{b,m}(k,y_k)\left|\Delta_{n,k}^rF(y)- \frac{\partial F}{\partial y_k}(y)\right| = f_{b,m}(k,y_k) \left|\Delta_{n,k}^rF(y)\right| \leq \frac{1}{2\sqrt{n}} \left\| \frac{\partial^2 F}{\partial y_k^2} \right\|_{\infty} \|f_{b,m}\|_{F, \infty}, 
\\ & g_{b,m}(k,y_k)\left|\Delta_{n,k}^lF(y)\mathbbm{1}_{\left\{y_k \geq \frac{1}{\sqrt{n}}\right\}}- \frac{\partial F}{\partial y_k}(y)\right| = 0.
\end{split}
\end{equation*}
We can argue similarly for the sum
\begin{equation*}
 \sum\limits_{k=1}^{N-1} \left( \Delta_{n,k-1}^r F(y -  \frac{e_k}{\sqrt{n}}) + \Delta_{n,k+1}^r F(y - \frac{e_k}{\sqrt{n}}) -2\Delta_{n,k}^l F(y) \right)y_k,
\end{equation*}
and find that
\begin{equation*}
\sup\limits_{y \in \frac{1}{\sqrt{n}} \mathbb{N}^{N-1}} \left| BF(y)- \sum\limits_{k=1}^{N-1} \left[ \Delta_{n,k-1}^r F(y - \frac{e_k}{\sqrt{n}}) + \Delta_{n,k+1}^r F(y - \frac{e_k}{\sqrt{n}}) -2\Delta_{n,k}^l F(y) \right] y_k \right| \rightarrow 0,
\end{equation*}
where 
\begin{equation*}
BF(y) := \sum\limits_{k=1}^{N-1} \left[ \frac{\partial F}{\partial y_{k-1}} + \frac{\partial F}{\partial y_{k+1}}  -2 \frac{\partial F}{\partial y_k} \right]y_k.
\end{equation*}
Recalling our notational convention that terms referring to the $0^{\textrm{th}}$ and $N^{\textrm{th}}$ queues are zero, we have that
\begin{equation*}
BF(y) = \sum\limits_{k=1}^{N-1} \frac{\partial F}{\partial y_k} \left( y_{k+1}+y_{k-1} -2y_k \right).
\end{equation*}
Putting these cases together gives that Theorem \ref{microstatic} holds.
\end{proof}

\subsection{Proof of Theorem \ref{microdynamic} }

The proof of Theorem \ref{microstatic} will be used as a basis for the proof of Theorem \ref{microdynamic}. We start by proving some continuity type results for certain maps, connecting some of the features of the microscopic models to their mesoscopic counterparts. These will be used in an inductive argument which will allow us to prove that 
\begin{multline*}
\left(\left(Z_{n,i}^b(nt)/\sqrt{n}\right)_{i=1}^{\infty}, \left(Z_{n,i}^a(nt)/\sqrt{n}\right)_{i=1}^{\infty}, (\tau^i_{n}/n)_{i=1}^{\infty}, (m^i_n)_{i=1}^{\infty} \right) \\ \implies ((X_i^b)_{i=1}^{\infty}, (X_i^a)_{i=1}^{\infty}, (\tau^i)_{i=1}^{\infty}, (m^i)_{i=1}^{\infty} )
\end{multline*} in law in $\mathbb{D}([0,\infty), \mathbb{R}^{N-1})^{\mathbb{N}} \times \mathbb{D}([0,\infty), \mathbb{R}^{N-1})^{\mathbb{N}} \times [0,\infty]^{\mathbb{N}} \times \mathbb{R}^{\mathbb{N}}$. Note that, for a metric space $M$, we use the topology of pointwise convergence for $M^{\mathbb{N}}$, which is itself metrizable. We then prove a continuity-type result for the map which sends these processes to their associated dynamic models. Finally, we conclude with an application of the Skorohod representation theorem. 

\begin{prop}\label{microtheta}
Fix some $m \in \mathbb{R}$. Let $w_n: [0,\infty) \rightarrow \mathbb{N}^{N-1} \times \mathbb{N}^{N-1}$ and $w:[0,\infty) \rightarrow \mathbb{R}^{N-1} \times \mathbb{R}^{N-1}$ be such that, for every $T>0$, 
 $$\sup\limits_{t \in [0,T]} \left| \frac{1}{\sqrt{n}}(w_n^1(nt),w_n^2(nt)) -(w^1(t),w^2(t))\right| \rightarrow 0.$$ Then we have that, for every $T>0$, 
$$\sup\limits_{t \in [0,T]} \left|n \theta^n_{u,m}(w_n(nt))- \theta_{u,m}(w(t)) \right|\rightarrow 0,$$ and $$\sup\limits_{t \in [0,T]} \left|n \theta^n_{d,m}(w_n(nt))- \theta_{d,m}(w(t)) \right| \rightarrow 0.$$ 
\end{prop}
\begin{proof}
This is a direct consequence of assumption (i) in Section 2.4.
\end{proof}

\begin{prop}\label{microregen}
Suppose $(Z_n^1,Z_n^2)$ is a sequence in $\mathscr{M}(\mathbb{N}^{N-1} \times \mathbb{N}^{N-1})$ such that $$P_n((Z_n^1,Z_n^2)) \implies (X^1,X^2)$$ in $\mathscr{M}((\mathbb{R}^+)^{N-1} \times (\mathbb{R}^+)^{N-1})$. Let $\tilde{R}_n: \mathscr{M}(\mathbb{N}^{N-1} \times \mathbb{N}^{N-1}) \times \{u,d\} \rightarrow \mathscr{M}(\mathbb{N}^{N-1} \times \mathbb{N}^{N-1})$ such that for $(\mu, k) \in \mathscr{M}(\mathbb{N}^{N-1} \times \mathbb{N}^{N-1}) \times \{ u,d \}$, 
\begin{equation*}
\tilde{R}_n(\mu, k)(A):= \int_{\mathbb{N}^{N-1} \times \mathbb{N}^{N-1}}  R_n(x_1,x_2,k)(A) \; \mu(\textrm{d}x_1, \textrm{d}x_2). 
\end{equation*}
Similarly, define $\tilde{R}: \mathscr{M}((\mathbb{R}^+)^{N-1} \times (\mathbb{R}^+)^{N-1} ) \times \{u,d\} \rightarrow \mathscr{M}((\mathbb{R}^+)^{N-1} \times (\mathbb{R}^+)^{N-1})$ such that for $(\nu, k) \in \mathscr{M}((\mathbb{R}^+)^{N-1} \times (\mathbb{R}^+)^{N-1} )  \times \{ u,d \}$,
\begin{equation*}
\tilde{R}(\nu, k)(B):= \int_{(\mathbb{R}^+)^{N-1} \times (\mathbb{R}^+)^{N-1}}  R(x_1,x_2,k)(A) \; \nu(\textrm{d}x_1, \textrm{d}x_2).
\end{equation*}
Then, for $k \in \{u,d\}$, $$\tilde{R}_n(Z_n^1,Z_n^2,k)\circ P_n^{-1} \implies \tilde{R}(X^1,X^2,k)$$ in $\mathscr{M}((\mathbb{R}^+)^{N-1} \times (\mathbb{R}^+)^{N-1})$.
\end{prop}
\begin{proof}
We apply the Skorohod representation theorem to the weak convergence of $P_n((Z_n^1,Z_n^2))$ to $(X^1,X^2)$, so we assume that $P_n((Z_n^1,Z_n^2)) \rightarrow (X^1,X^2)$ almost surely on some probability space $(\Omega, \mathscr{F}, \mathbb{P})$. We then have that $$R_n(Z_n^1,Z_n^2,k)\circ P_n^{-1} \rightarrow R(X^1,X^2,k)$$ in $\mathscr{M}((\mathbb{R}^+)^{N-1}) \times (\mathbb{R}^+)^{N-1})$ $\mathbb{P}$- almost surely. Let $A \in \mathscr{B}((\mathbb{R}^+)^{N-1} \times (\mathbb{R}^+)^{N-1})$ be a continuity set for the measure $\tilde{R}(X^1,X^2,k)$. We then have that $A$ is a continuity set for $R(X^1, X^2, k)$ $\mathbb{P}$- almost surely, from which it follows that 
\begin{equation*}
R_n(Z_n^1,Z_n^2,k)\circ P_n^{-1}(A) \rightarrow R(X^1,X^2,k)(A)
\end{equation*}
$\mathbb{P}$ - almost surely. Taking expectations then gives the result.
\end{proof}

\begin{prop}\label{microstopping}
Suppose that $f_n: [0,\infty) \rightarrow \mathbb{R}_{ >0  }$ and $f:[0,\infty) \rightarrow \mathbb{R}_{>0}$ are such that $f$ is continuous, there exists a constant $c >0$ such that $f \geq c$, and $f_n \rightarrow f$ uniformly on compact sets. Suppose also that $x_n \in \mathbb{R}$, $x \in \mathbb{R}$ such that $x_n \rightarrow x$. Define $$\tau_n:= \inf \left\{ t \geq 0 \; \big| \; x_n \leq \int_0^t f_n(s) \textrm{ds} \right\}$$ and 
$$\tau:= \inf \left\{ t \geq 0 \; \big| \; x \leq \int_0^t f(s) \textrm{ds} \right\}.$$
Then $\tau_n  \rightarrow \tau $.
\end{prop}
\begin{proof}
Fix $\delta >0$. We argue that, for $n$ large enough, $\tau_n \leq \tau + \delta.$ Note that 
$$\left| \left(x_n - \int_0^{\tau} f_n(s) \textrm{ds}\right) - \left( x- \int_0^{\tau} f(s) \textrm{ds}\right) \right| \leq |x_n-x| + \tau \times \sup\limits_{t \in [0,\tau]}|f_n(t)-f(t)|.$$ 
Since $f \geq c > 0$, we have that for $n$ large enough $\inf\limits_{t \in [0,\tau]} f_n(t) > \frac{c}{2}$. Also, for $n$ large enough, $$|x_n-x|+ \tau \times \sup\limits_{t \in [0,\tau]}|f_n(t)-f(t)| < \frac{\delta c}{2}.$$ Therefore, $$x_n - \int_0^{\tau} f_n(s) \textrm{ds} < x- \int_0^{\tau} f(s) \textrm{ds} + \frac{\delta c}{2} = \frac{\delta c}{2}$$ for large enough n. So $$x_n- \int_0^{\tau+ \delta} f_n(s) \textrm{ds} \leq x_n- \int_0^{\tau} f_n(s) \textrm{ds} - \frac{\delta c}{2} < 0.$$ Therefore $\tau_n \leq \tau + \delta$ for large enough n. Similarly, we have  $\tau_n \geq \tau - \delta$ for large enough n, concluding the proof.
\end{proof}

We can now prove the following result, whose proof constitutes the main step in proving Theorem \ref{microdynamic}.

\begin{prop}\label{microvectorconvergence}
Suppose that $$\left(\frac{Z_{n,1}^b(0)}{\sqrt{n}}, \frac{Z_{n,1}^a(0)}{\sqrt{n}} \right) \implies (X^b_1(0), X^a_1(0))$$ in law in $(\mathbb{R}^+)^{N-1} \times (\mathbb{R}^+)^{N-1}$. Let $(\hat{Z}_n^b(t),\hat{Z}^a_n(t),m_n(t))$ be dynamic microscopic models with initial data $(Z_{n,1}^b(0),Z_{n,1}^a(0),m^1)$, and let $(\hat{X}^b(t),\hat{X}^a(t),m(t))$ be the dynamic mesoscopic model with initial data $(X_1^b(0),X_1^a(0),m^1)$. Then 
\begin{multline*}
\left(\left(Z_{n,i}^b(nt)/\sqrt{n}\right)_{i=1}^{\infty}, \left(Z_{n,i}^a(nt)/\sqrt{n}\right)_{i=1}^{\infty}, ( \tau^i_{n}/n )_{i=1}^{\infty}, (m^i_n)_{i=1}^{\infty} \right) \\ \implies ((X_i^b)_{i=1}^{\infty}, (X_i^a)_{i=1}^{\infty}, (\tau^i)_{i=1}^{\infty}, (m^i)_{i=1}^{\infty} )
\end{multline*} in law in $\mathbb{D}([0,\infty), \mathbb{R}^{N-1})^{\mathbb{N}} \times \mathbb{D}([0, \infty), \mathbb{R}^{N-1})^{\mathbb{N}} \times [0,\infty)^{\mathbb{N}} \times \mathbb{R}^{\mathbb{N}}$. 
\end{prop}
\begin{proof}
We begin the proof by introducing the following notation. 
\begin{enumerate}[(i)]
\item $\mathscr{U}_n^{b,M}(t):= (Z_{n,i}^b(nt)/\sqrt{n})_{i=1}^{M}, \; \; \; \; \; \; \mathscr{U}^{b,M}(t):= (X_i^b(t))_{i=1}^M.$
\item $\mathscr{U}_n^{a,M}(t):= (Z_{n,i}^a(nt)/\sqrt{n})_{i=1}^{M}, \; \; \; \; \; \; \mathscr{U}^{a,M}(t):= (X_i^a(t))_{i=1}^M.$
\item $\tilde{\tau}_{n,u}^M:= (\tau_{n,u}^i/n)_{i=1}^M, \; \; \; \; \; \; \tilde{\tau}_u^M:= (\tau_u^i)_{i=1}^M.$
\item $\tilde{\tau}_{n,d}^M:= (\tau_{n,d}^i/n)_{i=1}^M, \; \; \; \; \; \; \tilde{\tau}_d^M:= (\tau_d^i)_{i=1}^M.$
\item $\tilde{m}_n^M:= (m_n^i)_{i=1}^M, \; \; \; \; \; \; \tilde{m}^M:= (m^i)_{i=1}^M.$

\end{enumerate}
We will prove by induction that for every $M \geq 1$, \begin{equation} \label{microvectorconv} (\mathscr{U}_n^{a,M}, \mathscr{U}_n^{b,M},\tilde{\tau}_{n,a}^{M-1}, \tilde{\tau}_{n,b}^{M-1}, \tilde{m}_n^M) \implies (\mathscr{U}^{a,M}, \mathscr{U}^{b,M}, \tilde{\tau}_{a}^{M-1}, \tilde{\tau}_{b}^{M-1}, \tilde{m}^M) \end{equation} in law in $\mathbb{D}([0,\infty), \mathbb{R}^{N-1})^M \times \mathbb{D}([0,\infty), \mathbb{R}^{N-1})^M \times [0,\infty)^{M-1} \times [0,\infty)^{M-1} \times \mathbb{R}^M$, from which the result follows. Note that by applying Theorem \ref{microstatic}, we have
$$(\mathscr{U}_n^{a,1},\mathscr{U}_n^{b,1}, \tilde{m}_n^1) \implies (\mathscr{U}^{a,1},\mathscr{U}^{b,1},\tilde{m}^1)$$ in $\mathbb{D}([0,\infty),\mathbb{R}^{N-1}) \times \mathbb{D}([0,\infty),\mathbb{R}^{N-1}) \times \mathbb{R}$. Suppose we had for some $M \geq 1$ that (\ref{microvectorconv}) holds. Then we have that
\begin{equation*}(\mathscr{U}_n^{a,M}, \mathscr{U}_n^{b,M},\tilde{\tau}_{n,a}^{M-1}, \tilde{\tau}_{n,b}^{M-1}, \tilde{m}_n^M, Y_{n,a}^M, Y_{n,b}^M) \implies (\mathscr{U}^{a,M},\mathscr{U}^{b,M},\tilde{\tau}_{a}^{M-1}, \tilde{\tau}_{b}^{M-1}, \tilde{m}^M, Y_a^M, Y_b^M).
\end{equation*}
By the Skorohod convergence theorem, we can assume that this convergence holds almost surely. Recall that 
\begin{equation*}
\tau_{n,u}^{M}:= \inf \left\{t \geq 0 \; \bigg| \; \int_0^t \theta_{u,m^{M}_n}^n(Z_{n,M}^b(s),Z_{n,M}^a(s))  \; \textrm{d}s \geq Y_{n,u}^{M} \right\},
\end{equation*} 
\begin{equation*}
\tau_{n,d}^{M}:= \inf \left\{t \geq 0  \; \bigg| \; \int_0^t \theta_{d,m^{M}_n}^n(Z_{n,M}^b(s),Z_{n,M}^a(s))  \; \textrm{d}s \geq Y_{n,d}^{M} \right\}.
\end{equation*}  By a change of variables, 
$$\int_0^t \theta_{u,m_n^M}^n(Z_{n,M}^b(s),Z_{n,M}^a(s) ) \textrm{d}s= \int_0^{t/n}n \theta_{u,m_n^M}^n(Z_{n,M}^b(nr),Z_{n,M}^a(nr) ) \textrm{d}r.$$ Therefore,
\begin{equation*} \begin{split} 
\frac{1}{n}\tau_{n,u}^{M}&= \frac{1}{n} \inf \left\{ t \geq 0 \; \bigg| \; \int_0^{t/n}n   \theta_{u, m_n^M}^n(Z_{n,M}^b(nr),Z_{n,M}^a(nr) ) \textrm{d}r \geq Y_{n,u}^{M} \right\} \\ & = \inf \left\{ t \geq 0 \; \bigg| \; \int_0^{t}n   \theta_{u,m_n^M}^n(Z_{n,M}^b(nr),Z_{n,M}^a(nr) ) \textrm{d}r \geq Y_{n,u}^{M} \right\}. 
\end{split} 
\end{equation*}
As $m_n^M$ and $m^M$ both take values in $m^1 + \epsilon \mathbb{Z}$ and $m_n^M \rightarrow m^M$ almost surely, we have that $m_n^M=m^M$ for large enough $n$ almost surely. Note that since the limits $X_M^a, X_M^b$ are continuous, convergence of $Z_{n,M}^a(nt)/\sqrt{n}$ to $X_M^a(t)$ and $Z_{n,M}^b(nt)/\sqrt{n}$ to $X_M^b(t)$ in $\mathbb{D}([0,\infty),\mathbb{R}^{N-1})$ implies that, for every $T>0$,  $$\sup\limits_{t \in [0,T]} \left|\frac{1}{\sqrt{n}}(Z_{n,M}^a(nt),Z_{n,M}^b(nt))-(X_M^a(t),X^b_M(t))\right| \rightarrow 0.$$ Since $\theta_{u,m}$ was assumed to be continuous and $X_M^b, X_M^a$ are continuous, $\theta_{u,m}(X_M^b(t),X_M^a(t))$ is continuous almost surely. We can therefore apply Propositions \ref{microtheta} and \ref{microstopping} to deduce that 
$$\frac{1}{n}\tau_{n,u}^{M} \rightarrow \inf \left\{ t \geq 0 \; \bigg| \; \int_0^t \theta_{u,m}(X_M^b(s),X_M^a(s))\;  \textrm{d}s \geq Y_u^{M} \right\}= \tau_u^{M}$$ almost surely. Similarly, $\frac{1}{n} \tau_{n,d}^{M} \rightarrow \tau_d^{M}$ almost surely. Note that $\tau_u^M \neq \tau_d^M$ almost surely. This implies that $m_n^{M+1} \rightarrow m^{M+1}$ almost surely. Therefore, we have deduced that 
\begin{equation} \label{microdynind}(\mathscr{U}_n^{a,M}, \mathscr{U}_n^{b,M},\tilde{\tau}_{n,a}^{M}, \tilde{\tau}_{n,b}^{M}, \tilde{m}_n^{M+1}) \implies (\mathscr{U}^{a,M},\mathscr{U}^{b,M},\tilde{\tau}_{a}^{M}, \tilde{\tau}_{b}^{M}, \tilde{m}^{M+1}).
\end{equation}
For $i=1,...M+1$, let $A_i \in \mathscr{B}(\mathbb{D}([0,\infty); \mathbb{R}^{N-1}))$ be a continuity set for the $\mathbb{D}([0,\infty); \mathbb{R}^{N-1})$- valued random variable $X^a_i$. Similarly, let $B_i$ be continuity sets for the random variables $X^b_i$, and let $p_i \in (m^1+ \epsilon\mathbb{Z} )$. For $i=1,...M$, let $C_{i},D_{i}$ be continuity sets for the random variables $\tau_u^i$ and $\tau_d^i$ respectively. Define:
$$\tilde{A}_K:= \prod\limits_{i=1}^K A_i, \; \; \; \; \; \; \tilde{B}_K:= \prod_{i=1}^K B_i, \; \; \; \; \; \; \tilde{C}_K:=\prod\limits_{i=1}^K C_i, \; \; \; \; \; \; \tilde{D}_K:=\prod\limits_{i=1}^K D_i, \; \; \; \; \; \; P_K:=\prod\limits_{i=1}^K \left\{ p_i \right\}.$$ Suppose that 
$$\mathbb{P}\left[(\mathscr{U}^{b,M}, \mathscr{U}^{a,M},\tilde{\tau}_{a}^{M}, \tilde{\tau}_{b}^{M}, \tilde{m}^{M+1}) \in \tilde{A}_M \times \tilde{B}_M \times \tilde{C}_M \times \tilde{D}_M \times P_{M+1} \right]> 0.$$ We then have that, for large enough $n$, 
\begin{equation*}
\mathbb{P}\left[(\mathscr{U}_n^{b,M}, \mathscr{U}^{a,M}_n,\tilde{\tau}_{n,a}^M, \tilde{\tau}_{n,b}^M, \tilde{m}_n^{M+1}) \in \tilde{A}_M \times \tilde{B}_M \times \tilde{C}_M \times \tilde{D}_M \times P_{M+1} \right] >0,
\end{equation*}
from which it follows that
\begin{equation*} 
\begin{split}
\mathbb{P}& \left[(\mathscr{U}_n^{b,M+1}, \mathscr{U}^{a,M+1}_n, \tilde{\tau}_{n,a}^{M}, \tilde{\tau}_{n,b}^{M}, \tilde{m}_n^{M+1}) \in \tilde{A}_{M+1} \times \tilde{B}_{M+1} \times \tilde{C}_{M} \times \tilde{D}_M \times P_{M+1} \right]
\\ =& \mathbb{P}\left[(\mathscr{U}_n^{b,M}, \mathscr{U}^{a,M}_n,\tilde{\tau}_{n,a}^M, \tilde{\tau}_{n,b}^M, \tilde{m}_n^{M+1}) \in \tilde{A}_M \times \tilde{B}_M \times \tilde{C}_M \times \tilde{D}_M \times P_{M+1} \right] 
\\ & \times \mathbb{P}_n\left[\frac{1}{\sqrt{n}}(Z_{n,M+1}^b(nt),Z^a_{n,M+1}(nt)) \in A_{M+1} \times B_{M+1} \right],
\end{split} 
\end{equation*}
where $\mathbb{P}_n$ denotes the conditional probability law of $\mathbb{P}$ given the event 
\begin{equation*}
(\mathscr{U}_n^{b,M}, \mathscr{U}^{a,M}_n,\tilde{\tau}_{n,a}^M, \tilde{\tau}_{n,b}^M, \tilde{m}_n^{M+1}) \in \tilde{A}_M \times \tilde{B}_M \times \tilde{C}_M \times \tilde{D}_M \times P_{M+1}.
\end{equation*}
By (\ref{microdynind}), we know that 
\begin{multline*}
\mathbb{P}\left[(\mathscr{U}_n^{b,M},\mathscr{U}^{a,M}_n,\tilde{\tau}_{n,a}^M, \tilde{\tau}_{n,b}^M, \tilde{m}_n^{M+1}) \in \tilde{A}_M \times \tilde{B}_M \times \tilde{C}_M \times \tilde{D}_M \times P_{M+1} \right] \\ \rightarrow \mathbb{P}\left[(\mathscr{U}^{b,M},\mathscr{U}^{a,M},\tilde{\tau}_{a}^M, \tilde{\tau}_{b}^M, \tilde{m}^{M+1}) \in \tilde{A}_M \times \tilde{B}_M \times \tilde{C}_M \times \tilde{D}_M \times P_{M+1} \right].
\end{multline*}
Let $\mathbb{Q}_n$ be the law on $(\mathbb{R}^+)^{N-1} \times (\mathbb{R}^+)^{N-1}$ induced by $R^n(Z_{n,M}^a(\tau^M_n),Z_{n,M}^b(\tau^M_n), k)\circ P_n^{-1}$ and the probability measure $\mathbb{P}_n$, where $k$ is the direction of the last price change. That is
\begin{equation*}
\mathbb{Q}_n(A)= \mathbb{E}^{\mathbb{P}_n} \left[ (R^n(Z_{n,M}^a(\tau^M_n),Z_{n,M}^b(\tau^M_n), k)\circ P_n^{-1}) (A) \right]
\end{equation*}
Similarly, we define the measure $\mathbb{Q}$ on $(\mathbb{R}^+)^{N-1} \times (\mathbb{R}^+)^{N-1}$ by setting
\begin{equation*}
\mathbb{Q}(A)= \mathbb{E}^{\tilde{\mathbb{P}}} \left[ R(X_{M}^a(\tau^M),X_{M}^b(\tau^M), k) (A) \right],
\end{equation*}
where $\tilde{\mathbb{P}}$ here is the conditional probability law given the event that 
\begin{equation*}
(\mathscr{U}^{b,M}, \mathscr{U}^{a,M},\tilde{\tau}_{a}^M, \tilde{\tau}_{b}^M, \tilde{m}^{M+1}) \in \tilde{A}_M \times \tilde{B}_M \times \tilde{C}_M \times \tilde{D}_M \times P_{M+1}.
\end{equation*}  
It follows from (\ref{microdynind}) and the continuity of $(X_M^b, X_M^a)$ that the law of $\frac{1}{\sqrt{n}}(Z_{n,M}^b(\tau^M_n),Z_{n,M}^a(\tau^M_n), k)$ under the measure $\mathbb{P}_n$ converges to the law of $(X_{M}^b(\tau^M),X_{M}^a(\tau^M), k)$ under $\tilde{\mathbb{P}}$. It then follows by an application of Proposition \ref{microregen} that $\mathbb{Q}_n \implies \mathbb{Q}$. We can now apply our result for convergence in a static setting, Theorem \ref{microstatic}, to deduce that  
\begin{equation*} 
\begin{split}
\mathbb{P}_n & \left[\frac{1}{\sqrt{n}}(Z_{n,M+1}^b(n \cdot),Z_{n,M+1}^a(n \cdot)) \in A_{M+1} \times B_{M+1} \right]  \rightarrow \tilde{\mathbb{P}}\left[(X_{M+1}^b, X_{M+1}^a) \in A_{M+1} \times B_{M+1}\right]. \end{split} 
\end{equation*}
We have therefore deduced that 
\begin{multline*} 
\mathbb{P} \left[(\mathscr{U}_n^{b,M+1},\mathscr{U}_n^{a,M+1},\tilde{\tau}_{n,a}^{M}, \tilde{\tau}_{n,b}^{M}, \tilde{m}_n^{M+1}) \in \tilde{A}_{M+1} \times \tilde{B}_{M+1} \times \tilde{C}_{M} \times \tilde{D}_M \times P_{M+1} \right] \\ \rightarrow \mathbb{P} \left[(\mathscr{U}^{b,M+1},\mathscr{U}^{a,M+1}, \tilde{\tau}_{a}^{M}, \tilde{\tau}_{b}^{M}, \tilde{m}^{M+1}) \in \tilde{A}_{M+1} \times \tilde{B}_{M+1} \times \tilde{C}_{M} \times \tilde{D}_M \times P_{M+1} \right]. 
\end{multline*}
So we have that $(\mathscr{U}_n^{b,M+1},\mathscr{U}_u^{a,M+1},\tilde{\tau}_{n,a}^{M}, \tilde{\tau}_{n,b}^{M}, \tilde{m}_n^{M+1}) \implies (\mathscr{U}^{b,M+1}, \mathscr{U}^{a,M+1},\tilde{\tau}_{a}^{M}, \tilde{\tau}_{b}^{M}, \tilde{m}^{M+1})$ which concludes our inductive argument.
\end{proof}

\begin{prop}\label{microconcat}
Let $g : \mathbb{D}([0,\infty); \mathbb{R}^{N-1})^{\mathbb{N}} \times [0, \infty)^{\mathbb{N}} \rightarrow \mathbb{D}([0,\infty); \mathbb{R}^{N-1})$ such that 
$$g((u_i)_{i=1}^{\infty}, (t_i)_{i=1}^{\infty})(t):= \sum\limits_{j=1}^{\infty} u_j\left( t- \sum\limits_{i=1}^{j-1} t_i \right) \mathbbm{1}_{ \left\{\sum\limits_{i=1}^{j-1} t_i \leq t < \sum\limits_{i=1}^j t_i \right\}}.$$
Suppose that 
\begin{enumerate}[(i)]
\item $t_i>0$ for every $i$.
\item $\sum\limits_{i = 1}^{\infty} t_i = \infty.$
\item $u_i$ is continuous for every $i$
\end{enumerate}
Then $g$ is continuous at the point $((u_i)_{i=1}^{\infty}, (t_i)_{i=1}^{\infty})$.
\end{prop}
\begin{proof}
Suppose that the sequence $((u_i^n)_{i=1}^{\infty}, (t_i^n)_{i=1}^{\infty})$ converges to $((u_i)_{i=1}^{\infty}, (t_i)_{i=1}^{\infty})$ in $\mathbb{D}([0,\infty) ; \mathbb{R}^{N-1})^{\mathbb{N}} \times [0,\infty)^{\mathbb{N}}$. It is sufficient to prove that 
\begin{equation*}
g((u_i^n)_{i=1}^{\infty}, (t_i^n)_{i=1}^{\infty}) \rightarrow g((u_i)_{i=1}^{\infty}, (t_i)_{i=1}^{\infty})
\end{equation*}
in $\mathbb{D}([0,T] ; \mathbb{R}^{N-1})$ for every $T \geq 0$ such that $g((u_i)_{i=1}^{\infty}, (t_i)_{i=1}^{\infty})$ is continuous at $T$. That is, it is enough to prove convergence in $\mathbb{D}([0,T] ; \mathbb{R}^{N-1})$ for every $T \geq 0$ such that 
\begin{equation*}
\sum\limits_{i = 1}^L t_i \neq T
\end{equation*}
for every $L$. Let $T \geq 0$ be such a point. By condition (ii), there exists $K \geq 1$ such that
\begin{equation}\label{K_sum}
\sum\limits_{i = 1}^K t_i > T.
\end{equation}
The functions $u_i$ are continuous, so we have by convergence of $u_i^n$ to $u_i$ in $\mathbb{D}([0,\infty) ; \mathbb{R}^{N-1})$ that $$\sup\limits_{t \in [0,T]} |u_i^n(t) -u_i(t)| \rightarrow 0$$ for each $i$. Recall that $f_n \rightarrow f$ in $\mathbb{D}( [0,T] ; M)$ for a metric space $M$ iff there exist continuous increasing bijections $\lambda_n: [0,T] \rightarrow [0,T]$ such that 
\begin{enumerate}[(a)]
\item $\sup\limits_{t \in [0,T]} |\lambda_n(t)-t| \rightarrow 0$ and 
\item $\sup\limits_{t \in [0,T]}|f_n(\lambda_n(t))-f(t)| \rightarrow 0$.
\end{enumerate}  
For our case, we define the continuous bijections $\lambda_n: [0,T] \rightarrow [0,T]$ such that for $m \geq 1$ we set $\lambda_n\left(\sum\limits_{i=1}^m t_i \right):= \sum\limits_{i=1}^m t_i^n$, and linearly interpolate between these points, together with the endpoints $\lambda_n(0)=0$ and $\lambda_n(T)=T$ (this may define the function for values greater than $T$ as well but we only care about the restriction). Then, for large enough $n$ so that the $t_i^n$ are all strictly positive, $\lambda_n$ is a strictly increasing continuous bijection from $[0,T]$ to $[0,T]$, and the $\lambda_n$ satisfy (a) above. Note that this can only be done if $T$ is a continuity point for $g$. We have that, for $t \in [0,T]$,
\begin{equation*}
\begin{split}
g((u_i^n)_{i=1}^{\infty}, (t_i^n)_{i=1}^{\infty})(\lambda_n(t))= & \sum\limits_{j=1}^{\infty} u_{j}^n\left(\lambda_n(t)-\sum\limits_{i=1}^{j-1}t_i^n \right) \mathbf{1}_{\left\{\sum\limits_{i=0}^{j-1}t_i^n \leq \lambda_n(t) < \sum\limits_{i=1}^j t_i^n \right\}} \\ = & \sum\limits_{j=1}^{\infty} u_{j-1}^n\left(\lambda_n(t)-\sum\limits_{i=0}^{j-1}t_i^n \right) \mathbf{1}_{\left\{\sum\limits_{i=0}^{j-1}t_i \leq t < \sum\limits_{i=0}^j t_i \right\}}.
\end{split}
\end{equation*}
It follows from (\ref{K_sum}) that 
\begin{equation*}
\begin{split}
& \sup\limits_{t \in [0,T]} | g((u_i^n)_{i=1}^{\infty}, (t_i^n)_{i=1}^{\infty})(\lambda_n(t)) - g((u_i)_{i=1}^{\infty}, (t_i)_{i=1}^{\infty})(t)|
\\ &   \leq \sup\limits_{j \leq K} \sup\limits_{t \in [0,T]} |u_{j-1}^n(t) - u_{j-1}(t)| \rightarrow 0.
\end{split}
\end{equation*}
Therefore (b) holds and we have the result.
\end{proof}
\begin{prop}\label{microprice}
Let $h: [0, \infty)^{\mathbb{N}} \times \mathbb{R}^{\mathbb{N}} \rightarrow \mathbb{D}([0,\infty); \mathbb{R})$ such that $$h((t_i)_{i=1}^{\infty}, (m_i)_{i=1}^{\infty}):= \sum\limits_{i=1}^{\infty} m_i \mathbbm{1}_{\left\{\sum\limits_{j=1}^{i-1} t_j \leq t < \sum\limits_{j=1}^i t_j \right\}}.$$ Suppose that 
\begin{enumerate}[(i)]
\item $t_i>0$ for every $i$.
\item $\sum\limits_{i = 1}^{\infty} t_i = \infty.$
\end{enumerate}
Then $h$ is continuous at the point $((t_i)_{i=1}^{\infty}, (m_i)_{i=1}^{\infty})$.
\begin{proof}
This is essentially the same as the proof of Proposition \ref{microconcat}.
\end{proof}
\end{prop}

We are now in a position to conclude the proof of Theorem \ref{microdynamic}. 

\begin{proof}[Proof of Theorem \ref{microdynamic}]
We begin by applying the Skorohod representation theorem to the weak convergence statement in Proposition \ref{microvectorconvergence}. This gives that 
\begin{multline*}
\left(\left(Z_{n,i}^b(nt)/\sqrt{n}\right)_{i=1}^{\infty}, \left(Z_{n,i}^a(nt)/\sqrt{n}\right)_{i=1}^{\infty}, (\tau^i_{n}/n)_{i=1}^{\infty}, (m^i_n)_{i=1}^{\infty} \right) \\ \rightarrow ((X_i^b)_{i=1}^{\infty}, (X_i^a)_{i=1}^{\infty}, (\tau^i)_{i=1}^{\infty}, (m^i)_{i=1}^{\infty} )
\end{multline*} 
in $\mathbb{D}([0,\infty); \mathbb{R}^{N-1})^{\mathbb{N}} \times \mathbb{D}([0,\infty); \mathbb{R}^{N-1})^{\mathbb{N}} \times [0,\infty)^{\mathbb{N}} \times \mathbb{R}^{\mathbb{N}}$ almost surely, where the vectors represent the key objects for the dynamic microscopic and mesoscopic models respectively. By Proposition \ref{microconcat}, we clearly have that $((X_i^a)_{i=1}^{\infty}, (\tau^i)_{i=1}^{\infty})$ and $((X_i^b)_{i=1}^{\infty}, (\tau^i)_{i=1}^{\infty})$ are continuity points of the map $g$ almost surely. Similarly, by Proposition \ref{microprice}, $((\tau^i)_{i=1}^{\infty}, (m^i)_{i=1}^{\infty})$ is a continuity point for the map $h$ almost surely. The result follows.   
\end{proof}

\subsection{Proof of Theorem \ref{mesostatic}}

As in the case of the proof of Theorem \ref{microstatic} we notice that the static mesoscopic systems decouple into two independent problems on either side of the mid. We therefore focus once more on proving the convergence on one side of the mid and, without loss of generality, choose to prove convergence of the bid side. The proof relies on an adaptation of Theorem 2.1 in T.Zhang \cite{Z}.

\begin{thm}[Theorem 2.1, \cite{Z}]\label{refconverge}
Let $(u, \eta)$ be a solution of the reflected stochastic heat equation (\ref{refSPDE}) with respect to a given white noise $W$, with initial data $u_0 \in C_0((0,1))^+$. Suppose that $\sigma$ and $h$ are Lipschitz in both variables and have linear growth in the second variable. For $n \geq 1$, let $(W^{n,k})_{k=1}^{n-1}$ be the independent family of Brownian motions given by 
\begin{equation*}
W^{n,k}_t := \sqrt{n}\left[W\left(t,\frac{k+1}{n}\right)-W\left(t, \frac{k}{n}\right) \right],
\end{equation*}
and let $u^{n}$ be the solution of the system of reflected SDEs 
\begin{equation*}\label{reflectedSDE2}
\textrm{d}u_t^{n,k}= \alpha n^2 \left( u_t^{n,k+1} +u_t^{n,k-1} -2u_t^{n,k} \right) \textrm{dt}+ h\left(k/n,u_t^{n,k}\right) \textrm{dt} + \sqrt{n}\sigma\left(k/n,u_t^{k,n}\right) \textrm{d}W_t^{n,k}+ \textrm{d}\eta_t^{n,k},
\end{equation*}
for $k=1,...,n-1$ with $u^{n,0}=u^{n,n}=0$ and initial data $u_0^{n} \in (\mathbb{R}^+)^{n-1}$. For each $t \geq 0$, define the function $u^n(t,x)$ by setting $u^n(t, \frac{k}{n}):=u^{n,k}_t$ and linearly interpolating between these points. Suppose that $$\sup\limits_{x \in [0,1]}|u^n(0,x)-u_0(x)| \rightarrow 0.$$ Then for every $p \geq 1$ and every $t \in[0,T]$
\begin{equation*}
\lim\limits_{n \rightarrow \infty} \mathbb{E}\left[\sup\limits_{t \in [0,T]} \sup\limits_{x \in[0,1]} |u^n(t,x)-u(t,x)|^p \right]= 0.
\end{equation*}
\end{thm}
The following corollary states that, as a consequence of Theorem \ref{refconverge}, we have convergence in law of the SDE system to the reflected SPDE in $C([0,\infty) ; C_0((0,1)))$ in the case where we have random initial data converging in law. The space $C([0,\infty) ; C_0((0,1)))$ is equipped with the topology of uniform convergence on compact sets in the time variable, so $f_n \rightarrow f$ in $C([0,\infty) ; C_0((0,1)))$ if and only if $f_n \rightarrow f$ in $C([0,T] ; C_0((0,1)))$ for every $T>0$. 
\begin{cor}\label{ref_converge}
Let $u$ be the solution of the reflected stochastic heat equation (\ref{refSPDE}) with initial data given by the law $\mu \in \mathscr{M}(C_0((0,1))^+)$. Suppose that $\sigma$ and $h$ are Lipschitz in both variables and have linear growth in the second variable. For $n \geq 1$, let $(W^{n,k})_{k=1}^{n-1}$ be an independent family of Brownian motions and let $u^{n}$ be the solution of the system of reflected SDEs 
\begin{equation*}\label{reflectedSDE2a}
\textrm{d}u_t^{n,k}= \alpha n^2 \left( u_t^{n,k+1} +u_t^{n,k-1} -2u_t^{n,k} \right) \textrm{dt}+ h\left(k/n,u_t^{n,k}\right) \textrm{dt} + \sqrt{n}\sigma\left(k/n,u_t^{k,n}\right) \textrm{d}W_t^{n,k}+ \textrm{d}\eta_t^{n,k},
\end{equation*}
for $k=1,...,n-1$ with $u^{n,0}=u^{n,n}=0$ and initial data given by the law $\nu_n \in \mathscr{M}((\mathbb{R}^+)^{n-1})$. For each $t \geq 0$, define the function $u^n(t,x)$ by setting $u^n(t,0) = 0$, $u^n(t,1) = 0$ and $u^n(t, \frac{k}{n}):=u^{n,k}_t$ for $k = 1,...,n-1$, and linearly interpolating between these points. Suppose that $(\nu_n \circ (\sqrt{n}Q_n)^{-1}) \implies \mu$ in law in $\mathscr{M}(C_0((0,1)))$. Then $u^n \implies u$ in law in $\mathscr{M}(C([0,\infty); C_0((0,1))))$.
\end{cor}
\begin{proof}
In the case where the initial data are deterministic, convergence in $\mathscr{M}(C([0,\infty) ; C_0((0,1)) )$ follows from convergence in $L^p(\Omega ; C([0,T] \times [0,1]))$ for every $T>0$. Turning to the case of random initial data, let $f \in C_b(C([0,\infty) ; C_0((0,1)))$. Define $\mu_n = \nu_n \circ (\sqrt{n} Q_n)^{-1}$. We want to prove that 
\begin{equation*}
\mathbb{E}^{\mu^n} \left[ f(u^n) \right] \rightarrow \mathbb{E}^{\mu} \left[ f(u) \right].
\end{equation*}
Let $F_n :C_0((0,1)) \rightarrow C_0((0,1))$ such that $F_n(u)(0) = F_n(u)(1) = 0$, $F_n(u)(i/n) = u(i/n)$ for $i = 1,...,n-1$, and $F_n(u)$ is linear in the intervals $[i/n, (i+1)/n]$ for $i = 0,1,...,n-1$. For $n \geq 1$ and $v \in C_0((0,1))$, define 
\begin{equation*}
g_n(v) := \mathbb{E}\left[ f(u^n) \; | \; u^n(0,\cdot) = F_n(v) \right].
\end{equation*}
Similarly, for $v \in C_0((0,1))$, define
\begin{equation*}
g(v) := \mathbb{E}\left[ f(u) \; | \; u(0,\cdot) = v \right].
\end{equation*}
Then we have that
\begin{equation*}
\mathbb{E}^{\mu_n} \left[ f(u^n) \right] = \int_{C_0((0,1))} g_n(v) \; \mu_n(\textrm{d} v), 
\end{equation*}
and 
\begin{equation*}
\mathbb{E}^{\mu} \left[ f(u) \right] = \int_{C_0((0,1))} g(v) \; \mu(\textrm{d}v), 
\end{equation*}
By the Skorohod representation theorem, we can realise the convergence of $\mu_n$ to $\mu$ with the random variables $u_0^n$ and $u_0$ on a common probability space $(\Omega, \mathscr{F}, \mathbb{P})$, so that $u_0^n \rightarrow u_0$ in $C_0((0,1))$ $\mathbb{P}$-almost surely. By the deterministic initial data case, we then have 
\begin{equation*}
g_n(u_0^n) \rightarrow g(u_0)
\end{equation*}
$\mathbb{P}$- almost surely. The result then follows by the DCT.
\end{proof}

\begin{proof}[Proof of Theorem \ref{mesostatic}]
Recall that the dynamics for the bid side of the $N^{\textrm{th}}$ static mesoscopic model are given by 
$$\textrm{d}X_N^{b,i}(t)= \alpha_b(X_N^{b, i+1}(t)+X_N^{b, i-1}(t)-2X_N^{b, i}(t)) \textrm{d}t+h_{b,m}^N(i,X_N^{b, i}(t)) \textrm{d}t + \sigma_{b,m}^N(i,X_N^{b, i}(t)) \textrm{d}W_t^{N,b,i}+ \textrm{d}\eta^{N,b,i}_t$$ for $i=1,...,N-1$ with $X_N^{b,0}=X_N^{b,N}=0$. It follows that 
\begin{multline}\label{reflectedSDEs}
\textrm{d}v_t^{N,b,i}= \alpha_b N^2 \left( v_t^{N,i+1} +v_t^{N,i-1} -2v_t^{N,i} \right) \textrm{d}t+ h_{b,m}\left(k/N,v_t^{N,i}\right) \textrm{d}t \\ + \sqrt{N}\sigma_{b,m}\left(i/N,v_t^{i,N}\right) \textrm{d}\tilde{W}_t^{N,b,i}+ \textrm{d}\eta_t^{N,b,i},
\end{multline} where we define $v_t^{N,b,i}:=\frac{1}{\sqrt{N}}X_N^{b,i}(N^2t)$ and $\tilde{W}_t^{N,b,i} = \frac{1}{N} W_{N^2 t}^{N,b,i}.$ The result then follows by an application of Corollary \ref{ref_converge}.
\end{proof}

\subsection{Proof of Theorem \ref{mesodynamic} }

As in the proof of Theorem \ref{microdynamic}, we build towards a proof by first presenting a series of smaller results. The arguments here reflect those made when upgrading the proof of Theorem \ref{microstatic} to a proof of Theorem \ref{microdynamic}. We therefore refer to the proof of Theorem \ref{microdynamic} to illustrate how to prove the key results in this section. In the same way as the proof of Theorem \ref{microdynamic}, the main part of the proof is showing that 
\begin{multline*}
((Q_N((X_{N,i}^b(N^2t, \cdot)))_{i=1}^{\infty},(Q_N((X_{N,i}^a(N^2t, \cdot)))_{i=1}^{\infty}, (\tau_{N,u}^i/N^2)_{i=1}^{\infty},(\tau_{N,d}^i/N^2)_{i=1}^{\infty}, (m_N^i)_{i=1}^{\infty}) \\ \implies ((u_i^b(t,\cdot))_{i=1}^{\infty},(u_i^a(t,\cdot))_{i=1}^{\infty}, (\tau_u^i)_{i=1}^{\infty}, (\tau_{d}^i)_{i=1}^{\infty}, (m^i)_{i=1}^{\infty})
\end{multline*}
in law in $\mathbb{D}([0,\infty) ;C_0((0,1)))^{\mathbb{N}} \times \mathbb{D}([0,\infty) ;C_0((0,1)))^{\mathbb{N}} \times [0,\infty)^{\mathbb{N}} \times [0,\infty)^{\mathbb{N}} \times \mathbb{R}^{\mathbb{N}}$. Once again, for a metric space $M$ we equip $M^{\mathbb{N}}$ with the topology of pointwise convergence, which is metrizable.

\begin{prop}\label{mesotheta}
Fix some $m \in \mathbb{R}$. Let $w_N:[0, \infty)  \rightarrow \mathbb{R}^{N-1}$ and $w: [0, \infty) \rightarrow C_0((0,1))$ be such that, for every $T>0$ $$\sup\limits_{t \in [0,T]} |Q_N(w_N(t))-w(t)| \rightarrow 0.$$ Then we have that, for every $T>0$,  $$\sup\limits_{t \in [0,T]} \left|N^2 \theta^N_{b,Nm}(w_N(t))- \theta_{b,m}(w(t)) \right|\rightarrow 0,$$ and $$\sup\limits_{t \in [0,T]} \left|N^2 \theta^N_{a,Nm}(w_N(t))- \theta_{a,m}(w(t)) \right| \rightarrow 0.$$ 
\end{prop}
\begin{proof}
This is a direct application of assumption (i) in Section 3.6.
\end{proof}

\begin{prop}\label{mesoregen}
Suppose $(X_N^1,X_N^2)$ is a sequence in $\mathscr{M}(\mathbb{R}^{N-1} \times \mathbb{R}^{N-1})$ such that $$(Q_N(X_N^1),Q_N(X_N^2)) \implies (u^1,u^2)$$ in $\mathscr{M}(C_0((0,1)) \times C_0((0,1)))$. Define $\tilde{R}^N: \mathscr{M}(\mathbb{R}^{N-1} \times \mathbb{R}^{N-1}) \times \{u,d\} \rightarrow \mathscr{M}(\mathbb{R}^{N-1} \times \mathbb{R}^{N-1})$ such that, for $(\mu, k) \in \mathscr{M}(\mathbb{R}^{N-1} \times \mathbb{R}^{N-1}) \times \{u, d \}$,
\begin{equation*}
\tilde{R}^N(\mu, k)(A):= \int_{\mathbb{R}^{N-1} \times \mathbb{R}^{N-1}}  R^N(x_1, x_2, k)(A) \; \mu(\textrm{d}x_1, \textrm{d}x_2).
\end{equation*}
Similarly, define $\tilde{R}: \mathscr{M}(C_0((0,1)) \times C_0((0,1)) \times \{u,d\} \rightarrow \mathscr{M}(C_0((0,1)) \times C_0((0,1)) )$ such that for $(\nu, k) \in \mathscr{M}(C_0((0,1)) \times C_0((0,1))) \times \{u, d \}$,
\begin{equation*}
\tilde{R}(\nu, k)(B):= \int_{C_0((0,1)) \times C_0((0,1))} R(u^1,u^2,k)(B) \; \nu(\textrm{d}x_1, \textrm{d}x_2).
\end{equation*}
Then for $k \in \{ u, d \}$, $$\tilde{R}^N((X^1,X^2), k)\circ Q_N^{-1} \implies \tilde{R}((u^1,u^2), k)$$ in $\mathscr{M}(C_0((0,1)) \times C_0((0,1)))$.
\end{prop}
\begin{proof}
This is essentially the same as the proof of Proposition \ref{microregen}.
\end{proof}

\begin{thm}\label{mesodynind}
Suppose that $$(Q_N(X_{N,1}^b(0)),Q_N(X_{N,1}^a(0))) \implies (u^b(0),u^a(0))$$ in law in $C_0((0,1))^+ \times C_0((0,1))^+$. Let $(X_N^b(t), X_N^a(t),m_N(t))$ be dynamic microscopic models with initial data $(X_{N,1}^b(0),X_{N,1}^a(0),m^1)$, and let $(u^b(t),u^a(t),m(t))$ be the dynamic macroscopic model with initial data $(u^b(0),u^a(0),m^1)$. Then 
\begin{multline}
((Q_N(X_{N,i}^b(N^2t)))_{i=1}^{\infty},(Q_N(X_{N,i}^a(N^2t)))_{i=1}^{\infty}, (\tau^i_{N,b}/N^2)_{i=1}^{\infty}, (\tau^i_{N,a}/N^2)_{i=1}^{\infty}, (m^i_N)_{i=1}^{\infty}) \\ \implies ((u_i^b)_{i=1}^{\infty},(u_i^a)_{i=1}^{\infty}, (\tau^i_{b})_{i=1}^{\infty}, (\tau^i_{a})_{i=1}^{\infty}, (m^i)_{i=1}^{\infty}
\end{multline} in law in $\mathbb{D}([0,\infty); C_0((0,1)))^{\mathbb{N}} \times \mathbb{D}([0,\infty) ; C_0((0,1)))^{\mathbb{N}} \times [0,\infty)^{\mathbb{N}} \times [0,\infty)^{\mathbb{N}} \times \mathbb{R}^{\mathbb{N}}$. 
\end{thm}
\begin{proof}
This follows the proof of Theorem \ref{microvectorconvergence} and we refer to this for the details. We once again use an inductive argument, and show that for every $M \geq 1$, 
\begin{multline}
((Q_N(X_{N,i}^b(N^2t)))_{i=1}^{M},(Q_N(X_{N,i}^a(N^2t)))_{i=1}^{M}, (\tau^i_{N,b}/N^2)_{i=1}^{M-1}, (\tau^i_{N,a}/N^2)_{i=1}^{M-1}, (m^i_N)_{i=1}^{M}) \\ \implies ((u_i^b)_{i=1}^{M},(u_i^a)_{i=1}^{M}, (\tau^i_{b})_{i=1}^{M-1}, (\tau^i_{a})_{i=1}^{M-1}, (m^i)_{i=1}^{M})
\end{multline}
in law in $\mathbb{D}([0,\infty); C_0((0,1)))^{\mathbb{N}} \times \mathbb{D}([0,\infty); C_0((0,1)))^{\mathbb{N}} \times [0,\infty)^{\mathbb{N}} \times [0,\infty)^{\mathbb{N}} \times \mathbb{R}^{\mathbb{N}}$, from which the result follows. Given the inductive hypothesis, we use the Skorohod representation theorem and Propositions \ref{mesotheta} and \ref{microstopping} to obtain convergence of the next rescaled stopping times in the sequence. It then follows that the next boundary position in the sequence also converges. By conditioning as in the proof of Theorem \ref{microdynind} and making use of Proposition \ref{mesoregen}, we can obtain convergence of the process up to the next stopping time, concluding the inductive argument.
\end{proof}
\begin{prop}\label{mesoconcat}
Let $g : \mathbb{D}([0,\infty) ; C_0((0,1)))^{\mathbb{N}} \times [0, \infty)^{\mathbb{N}} \rightarrow \mathbb{D}([0,\infty) ; C_0((0,1)))$ such that 
$$g((u_i)_{i=1}^{\infty}, (t_i)_{i=1}^{\infty})(t):= \sum\limits_{j=1}^{\infty} u_j\left( t- \sum\limits_{i=1}^{j-1} t_i \right) \mathbbm{1}_{ \left\{\sum\limits_{i=1}^{j-1} t_i \leq t < \sum\limits_{i=1}^j t_i \right\}}.$$
Suppose that 
\begin{enumerate}[(i)]
\item $t_i>0$ for every $i$.
\item $\sum\limits_{i = 1}^{\infty} t_i = \infty.$
\item $u_i$ is continuous for every $i$
\end{enumerate}
Then $g$ is continuous at the point $((u_i)_{i=1}^{\infty}, (t_i)_{i=1}^{\infty})$.
\end{prop}
\begin{proof}
This is essentially the same as the proof of Theorem \ref{microconcat}.
\end{proof}

\begin{proof}[Proof of Theorem \ref{mesodynamic}]
We can now conclude the proof of Theorem \ref{mesodynamic}. We begin by Skorohod representing the convergence in Theorem \ref{mesodynind}. So we have that 
\begin{multline}((Q_N(X_{N,i}^a(N^2t)))_{i=1}^{\infty},(Q_N(X_{N,i}^b(N^2t)))_{i=1}^{\infty}, (\tau^i_{N,a}/N^2)_{i=1}^{\infty}, (\tau^i_{N,b}/N^2)_{i=1}^{\infty}, (m^i_N)_{i=1}^{\infty}) \\ \rightarrow ((u_i^a)_{i=1}^{\infty},(u_i^b)_{i=1}^{\infty}, (\tau^i_{a})_{i=1}^{\infty}, (\tau^i_{b})_{i=1}^{\infty}, (m^i)_{i=1}^{\infty}) 
\end{multline}  
in $\mathbb{D}([0,\infty); C_0((0,1)))^{\mathbb{N}} \times \mathbb{D}([0,\infty); C_0((0,1)))^{\mathbb{N}} \times [0,\infty)^{\mathbb{N}} \times [0, \infty)^{\mathbb{N}} \times \mathbb{R}^{\mathbb{N}}$ almost surely. By Proposition \ref{mesoconcat}, we clearly have that $((u_i^a)_{i=1}^{\infty}, (\tau^i)_{i=1}^{\infty})$ and $((u_i^b)_{i=1}^{\infty}, (\tau^i)_{i=1}^{\infty})$  are continuity points of the map $g$ almost surely. Similarly, by Proposition \ref{microprice}, $((\tau^i)_{i=1}^{\infty}, (m^i)_{i=1}^{\infty})$ is a continuity point for the map $h$ almost surely. The result follows.  
\end{proof}


\begin{thebibliography}{99}

\bibitem{AJ} F. Abergel and A. Jedidi, \textit{A mathematical approach to order book modeling}, Int. J. Theor. Appl. Finance, {\bf 16}, 1350025, 40 pp. 2013.

\bibitem{AC} R. Almgren and N. Chriss, \textit{Optimal execution of portfolio transactions}, J. Risk {\bf 3}, 5--39, 2001.

\bibitem{AR} M. Avellaneda, J. Reed and S. Stoikov, \textit{Forecasting prices from level-I quotes in the presence of hidden liquidity}, Algorithmic Finance, {\bf 1}, 35--43, 2011.

\bibitem{BHQ} C. Bayer, U. Horst and J. Qiu, \textit{A functional limit theorem for limit order books with state dependent price dynamics}, Ann. Appl. Probab., {\bf 27}, 2753--2806, 2017.

\bibitem{BC} J. Blanchet and X. Chen, \textit{Continuous-time modeling of bid-ask spread and price dynamics in limit order books}, arXiv preprint, arXiv:1310.1103, 2013.

\bibitem{BoCe} A. Bovier and J. \u{C}ern\'{y}, \textit{Hydrodynamic limit for the $A+B\rightarrow\emptyset$ model}, Markov Proc. Related Fields, {\bf 13}, 
543--564, 2007.

\bibitem{CL} R. Cont and A. de Larrard, \textit{Price dynamics in a Markovian limit order market}, SIAM J. Fin. Math., {\bf 4}, 1--25, 2011.

\bibitem{CL2} R. Cont and A. de Larrard, \textit{Order book dynamics in liquid markets: limit theorems and diffusion approximations}, arXiv preprint, arXiv:1202.6412, 2012. 

\bibitem{CS} R. Cont, S. Stoikov and R. Talreja, \textit{A stochastic model for order book dynamics}, Oper. Res., {\bf 58}, 549--563, 2010.



\bibitem{DG} S. Deng and X. Gao, \textit{Hydrodynamic limit of order book dynamics}, Probab. Eng. Inform. Sciences, {\bf 32}, 96--125, 2018.


\bibitem{D} A. Diop, \textit{Sur la discr\'{e}tisation et le comportement \`{a} petit bruit d'EDS multidimensionnelles
dont les coefficients sont \`{a} d\'{e}riv\'{e}es singuli\`{e}res}, Doctoral dissertation, INRIA, 2003.

\bibitem{DP} C. Donati-Martin and E. Pardoux, \textit{White noise driven SPDEs with reflection}, Probab. Theory Relat. Fields, {\bf 95}, 1--24, 1993.

\bibitem{EK} S. Ethier and T. Kurtz, \textit{Markov processes: characterization and convergence}, Wiley Series in Probability and Statistics, 1986.

\bibitem{GP} M. Gould, M. Porter, S. Williams, M. McDonald, D. Fenn and S. Howison, \textit{Limit order books}, Quant. Finance, {\bf 13}, 1709--1742, 
2013.

\bibitem{HP} U. Horst and M. Paulsen, \textit{A law of large numbers for limit order books}, Math.  Oper. Res., { \bf 42}, 1280--1312, 2017.

\bibitem{HL} W. Huang, C-A. Lehalle and M. Rosenbaum, \textit{Simulating and analyzing order book data: the queue-reactive model}, J. Amer. Stat. Assoc., 
{\bf 110}, 107--122, 2015.

\bibitem{KRM} M. Keller-Ressel and M. M\"{u}ller, \textit{A Stefan-type stochastic moving boundary problem},  Stoch. Partial Differ. Equ. Anal. Comput. {\bf 4}, 746--790, 2016.

\bibitem{KSZ} K. Kim, R. Sowers and Z. Zheng, \textit{A stochastic Stefan problem}, J. Theor. Probab., {\bf 25}, 1040--1080, 2012.

\bibitem{K} L. Kruk, \textit{Functional limit theorems for a simple auction}, Math. Oper. Res., {\bf 28}, 716--751, 2003.
 
\bibitem{AK} A. Kyle, \textit{Continuous auctions and insider trading}, Econometrica, {\bf 53}, 1315--1335, 1985.

\bibitem{LRS} P. Lakner, J. Reed and S. Stoikov, \textit{High frequency asymptotics for the limit order book}, Market Micro. Liquidity, {\bf 2}, 165004, 2016.


\bibitem{M} H. Mendelson, \textit{Market behavior in a clearing house}, Econometrica, {\bf 50}, 1505--1524, 1982.

\bibitem{Mu} M. M\"{u}ller, \textit{A stochastic Stefan-type problem under first order boundary conditions}, Arxiv preprint arXiv:1601.03968, 2016.

\bibitem{MT} I. Muni Toke, \textit{The order book as a queueing system: average depth and influence of the size of limit orders}, Quantitative Finance, {\bf 15}, 795--808, 2015.

\bibitem{R} I. Ro\c{s}u, \textit{Liquidity and information in order driven markets}, SSRN Electronic Journal, 2015.

\bibitem{T} H. Tanaka, \textit{Stochastic Differential Equations with Reflecting Boundary Condition in Convex Regions}, Hiroshima Math. J. {\bf 9}, 163--177, 1979.


\bibitem{XZ} T. Xu and T. Zhang, \textit{White noise driven SPDEs with reflection: existence, uniqueness, and large deviation principles}, Stoch. Proc. Applic.
{\bf 119}, 3453--3470, 2009.

\bibitem{Z}
T. Zhang, \textit{Lattice Approximations of Reflected Stochastic Partial Differential Equations Driven by Space-Time White Noise}, Ann. Appl. Probab., 
{\bf 26}, 3602--3629, 2016. 

\bibitem{Zheng} Z. Zheng,  \emph{Stochastic Stefan Problems: Existence, Uniqueness and Modeling of Market Limit Orders}, PhD thesis, Graduate College of the University of Illinois at Urbana-Champaign, 2012.

\end{thebibliography}
\end{document}